\pdfoutput=1
\documentclass[a4paper]{article}

\usepackage[top=2.5cm, bottom=2.5cm, left=3.7cm, right=3.7cm]{geometry}	

\usepackage{amsmath, amsfonts, bm, amssymb, mathtools, bm, commath, amsthm}
\usepackage{doi}
\usepackage{graphicx}
\usepackage{wrapfig}
\usepackage{color}
\usepackage[dvipsnames]{xcolor}
\usepackage[font=small,labelfont=bf]{caption}
\usepackage{enumitem}
\usepackage{chngcntr}
\usepackage{bm} 
\usepackage{algorithmic,algorithm}
\usepackage{cancel}
\usepackage[maxbibnames=99,backend=biber]{biblatex}
\newtheorem{theorem}{Theorem}
\newtheorem{proposition}[theorem]{Proposition}
\newtheorem{lemma}[theorem]{Lemma}
\newtheorem{assumption}{Assumption}
\newtheorem{definition}{Definition}
\newtheorem{corollary}[theorem]{Corollary}

\theoremstyle{definition}
\newtheorem{remark}{Remark}
\newtheorem{example}{Example}

\usepackage[textwidth=1.8cm, textsize=scriptsize]{todonotes}
\newfloat{routine}{tbp}{lor} 
\floatname{routine}{Subroutine}          
\makeatother

\renewcommand{\cal}[1]{\mathcal{#1}}

\renewcommand{\r}{\mathbb{R}}

\newcommand{\n}{\mathbb{N}}

\newcommand{\Ebb}[1]{\mathbb{E}\left[#1\right]}

\newcommand{\iprod}[2]{\left\langle{#1},{#2}\right\rangle}
\usepackage[symbol]{footmisc}
\newcommand{\Xspace}{\mathsf{E}}

\newcommand{\tvdist}[1]{\norm{#1}_{tv}}
\newcommand{\oscnorm}[1]{\norm{#1}_{\textup{osc}}}
\newcommand{\mtm}[1]{P_{#1}}
\newcommand{\ir}[1]{\widetilde{P}_{#1}}
\newcommand{\ideal}{P_\infty}

\newcommand{\alphid}{\alpha_\infty}
\newcommand{\alphir}[1]{\widetilde{\alpha}_{#1}}
\newcommand{\alphmtm}[1]{\alpha_{#1}}

\DeclareMathOperator*{\ess}{ess}

\hypersetup{
	colorlinks   = true, 
	urlcolor     = {green!50!black}, 
	linkcolor    = {green!50!black}, 
	citecolor   = {green!50!black} 
}

\title{Analysis of Multiple-try Metropolis via \\ Poincar{\'e} inequalities} 
\author{Rocco Caprio \and Sam Power \and Andi Q. Wang}
\begin{document}
\maketitle

\begin{abstract}
    We study the Multiple-try Metropolis algorithm using the framework of Poincar{\'e} inequalities. We describe the Multiple-try Metropolis as an auxiliary variable implementation of a resampling approximation to an ideal Metropolis--Hastings algorithm. Under suitable moment conditions on the importance weights, we derive explicit Poincar{\'e} comparison results between the Multiple-try algorithm and the ideal algorithm. We characterize the spectral gap of the latter, and finally in the Gaussian case prove explicit non-asymptotic convergence bounds for Multiple-try Metropolis by comparison.
	\\
	
\end{abstract}

\section{Introduction} \label{sec:intro}

Markov chain Monte Carlo (MCMC) methods are one of the most fundamental tools for Bayesian computation and beyond, as a tool to sample from a posterior distribution of interest, known up to a normalizing constant. These methods are based on the construction of a Markov chain having the prescribed target distribution as invariant. Denote by $(\Xspace,\mathcal{F})$ an underlying Polish space with its Borel $\sigma$-algebra, and let $\pi$ denote the target distribution of interest. The \emph{Metropolis--Hastings} algorithm \cite{metropolis1953,hastings1970} constructs such a Markov chain as follows: Given a current position $x\in\Xspace$, one draws $Y\sim q(x,\cdot)$, where $q$ is a `proposal' Markov kernel from $\Xspace$ to itself, which we use to explore the space. Assuming the existence of appropriate density functions, with probability
\begin{equation*}
	\alpha (x,y) :=\min\{1,r(x,y)\}, \quad r(x,y):=\frac{\pi(y)\cdot q(y,x)}{\pi(x)\cdot q(x,y)}
\end{equation*}
one moves to $Y$, and otherwise remains at $x$. This mechanism produces a $\pi$-invariant Markov chain that can be shown to be ergodic under general conditions. 

The Multiple-try Metropolis (MTM) algorithm \cite{liu2000,frenkel1996} is a sampling method where one introduces $n$ different proposal samples, and a selection mechanism based on a weight function $w:\Xspace\times\Xspace\mapsto \r^+$. The motivations behind this method are twofold. From one side, one would expect that different proposal distributions lead to a better exploration of the space, and thus to faster mixing properties of the chain. On the other hand, modern computing architectures easily allow for sampling and evaluating multiple proposals in parallel. This type of parallelization is to be contrasted with the one immediately available from vanilla Metropolis algorithms, whereby one can simply run $n$ chains in parallel. Running multiple independent chains in parallel does not, however, reduce the burn-in required for each chain to equilibrate. As a result, the Multiple-try method has the potential to result in significant gains in terms of \emph{non-asymptotic} computational efficiency; see \cite[Section 1.2]{gagnon2023}.
In the next section we introduce in more detail the Multiple-try Metropolis algorithm, by viewing it an \textit{auxiliary-variable implementation} of an \textit{importance resampling approximation} to an `\textit{ideal}' MCMC algorithm, which motivates our analysis; we then present our main results.
 
\subsection{Multiple-try Metropolis and limiting kernels} \label{sec:mtmintro}
The mixing properties of the Metropolis--Hastings scheme are inherently related to the proposal distribution $q$ and its relation to the target distribution $\pi$. Intuitively, if we are able to incorporate information regarding the target into the proposal, we might expect that the resulting Metropolis--Hastings chain would have better mixing properties. One potential way to include this information is by defining a suitable \textit{weight function} $w$ depending on $\pi$, and then defining the proposal $q^w$:
\begin{align}
	q^w(x,\dif y)&:= \frac{q(x,\dif y)\cdot w(x,y)}{(qw)(x)},\quad \text{where} \quad (qw)(x):=\int q(x,\dif y) \cdot w(x,y).
 \label{eq:qw}
\end{align}
Practical implementation of this chain requires the ability to sample from $q^w$ and evaluate $(qw)(x)$. We will call this the \emph{ideal} Metropolis--Hastings kernel and write $\ideal$ for its transition kernel.
\begin{algorithm}[H]
	\begin{algorithmic}[1]
		\STATE{Input: current state $x$}
		\STATE{Draw $Y|x \sim q^w(x,\dif y)$.}
		\STATE{With probability
			\begin{equation*}
				\alphid(x,y):=\min\bigg\{1,\frac{\pi(y)\cdot q(y,x)\cdot w(y,x)\cdot (qw)(x)}{\pi(x)\cdot q(x,y)\cdot w(x,y)\cdot (qw)(y)}\bigg\}
			\end{equation*}
			move to $Y$; otherwise remain at $x$.}
	\end{algorithmic}
	\caption{$P_\infty$: Ideal Metropolis algorithm }
	\label{alg:ideal}
\end{algorithm}

Typically, however, $q^w$ is neither available to sample from, nor to evaluate in closed form. In this setting, one can imagine an \textit{importance resampling} approximation to $q^w$. One can implement this as follows:
\begin{routine}[H]
	\begin{algorithmic}[1]
		\STATE{Input: current state $x$}
        \STATE{For $i\in [n]$, sample $Y_i \sim q(x,\dif y_i)$;}
        \STATE{Compute $$(\widehat{qw}_n)(x,Y^{[n]}):=\frac{1}{n}\sum_{i=1}^n w(x,Y_i),$$ where $Y^{[n]}:=(Y_1,\dots,Y_n)$;}
        \STATE{Sample $I\sim\mathrm{Categorical}\bigg(\bigg\{\frac{w(x,Y_i)}{n\cdot (\widehat{qw}_n)(x,Y^{[n]})}:i\in[n]\bigg\}\bigg)$, and set $Y=Y_I$;}
        \STATE{Return $I$, proposal $Y$ and $Y^{[n]}$.}
	\end{algorithmic}
	\caption{$\tilde q^w_n:$ Importance resampling approximation to $q^w$}
	\label{alg:IR_approx}
\end{routine} 
One can write down the marginal law of the effective proposal $Y$ from Subroutine~\ref{alg:IR_approx} as
\begin{align}
    \widetilde{q}^w_n(x,\dif y)&:=\frac{q(x,\dif y)\cdot w(x,y)}{(\widetilde{qw})_n(x,y)};
	\quad \text{where} \nonumber \\ (\widetilde{qw})_n(x,y)&:=\Ebb{(\widehat{qw}_n)(x,Y^{[n]})^{-1}|Y_1=y}^{-1}.
 \label{eq:qwn}
\end{align}
with $Y^{[n]}|Y_1=y\sim \delta_y(\dif y_1)\prod_{i=2}^n q(x,\dif y_i)$. By application of a law of large numbers, it holds that
\begin{align*}
    (\widetilde{qw})_n(x,y) \rightarrow (qw)(x); \quad \text{and} \quad \widetilde{q}^w_n(x,\dif y) \rightarrow q^w(x,\dif y) \quad \text{as } n\rightarrow \infty,
\end{align*}
in a suitable sense (to be made precise later), noting in particular that the denominator of $\widetilde{q}^w_n(x,\dif y)$ should depend on $y$ only weakly as $n$ grows.

As such, one could envision sampling from $\pi$ by using a Metropolis--Hastings method with target measure $\pi$ and proposal kernel $\widetilde{q}^w_n$, i.e.

\begin{algorithm}[H]
	\begin{algorithmic}[1]
		\STATE{Input: current state $x$}
	    \STATE{Draw $Y\sim \tilde q^w_n(x,\dif y)$ as in Subroutine~\ref{alg:IR_approx};}
		\STATE{With probability
			\begin{equation*}
				\alphir{n}(x,Y):=\min\bigg\{ 1,\frac{\pi(Y)\cdot q(Y,x)\cdot w(Y,x) \cdot (\widetilde{qw}_n)(x,Y)}{\pi(x)\cdot q(x,Y)\cdot w(x,Y) \cdot (\widetilde{qw}_n)(Y,x)}\bigg\},
			\end{equation*}
			move to $Y$; otherwise remain at $x$.}
	\end{algorithmic}
	\caption{$\widetilde P_n$: Semi-ideal Metropolis algorithm}
	\label{alg:semiid}
\end{algorithm}
We call the resulting Markov kernel the \emph{semi-ideal} chain, denoted $\ir{n}$. While we are now able to sample from the proposal, the problem of the intractability of its normalizing constant $(\widetilde{qw}_n)(x,y)$ within the acceptance probability persists. 
The Multiple-try Metropolis strategy of \cite{liu2000} can be seen as a solution to this problem. In this algorithm, one introduces auxiliary variables to derive a structurally similar kernel $\mtm{n}$ which admits an implementable accept-reject step, and which approximates well (in the large $n$ limit) both $\ir{n}$ and $\ideal$.

In particular, one can simulate from the Multiple-try kernel $\mtm{n}$ as follows:
\begin{algorithm}[H]
	\begin{algorithmic}[1]
		\STATE{Input: current state $x$}
		\STATE{Draw $Y\sim \tilde q^w_n(x,\dif y)$ and obtain $Y^{[n]}, I$ as in Subroutine~\ref{alg:IR_approx};}
		\STATE{Draw $Z_i|Y\sim q(y,\dif z)$ for $i\in[n]\backslash \{I\}$.}
		\STATE{Set $Z_I=x$, and compute $(\widehat{qw}_n)(Y,Z^{[n]}):=n^{-1}\sum_{i=1}^n w(Y,Z_i)$}
		\STATE{With probability
			\begin{equation*}
				\alphmtm{n}(x,Y^{[n]},Z^{[n]}):=\min\bigg(1,\frac{\pi(Y)\cdot q(Y,x)\cdot w(Y,x) \cdot (\widehat{qw}_n)(x,Y^{[n]})}{\pi(x)\cdot q(x,Y)\cdot w(x,Y) \cdot (\widehat{qw}_n)(Y,Z^{[n]})}\bigg),
			\end{equation*}
			move to $Y$; otherwise remain at $x$.}
	\end{algorithmic}
	\caption{$P_n$: Multiple-try Metropolis algorithm}
	\label{alg:mtm}
\end{algorithm}
The auxiliary samples $Z^{[n]}$ in Algorithm~\ref{alg:mtm} are known as \textit{balancing trials} or \textit{shadow samples} \cite{bedard2012,yang2023}. This approach generally introduces an extra computational burden (we now require $2 n$ simulations from $q$, rather than $n$), but the sampler is implementable and produces a genuinely $\pi$-reversible Markov chain.  

Table \ref{table:kernels} summarizes the various kernels of interest in this work that will enter into the analysis, and whether the acceptance criterion depends only on the state, or also on the auxiliary variables, and whether the method is implementable. 
\begin{table}[] 
\centering
\begin{tabular}{|ll|l|l|l|}
\hline
\textbf{Markov kernel} & \textbf{}         & \textbf{Proposal}    & \textbf{Acceptance} & \textbf{Implementable} \\ \hline
MH                     & $P^{\textup{MH}}$ & $q$                  & State only          & Yes                    \\
Ideal MH (Alg.~\ref{alg:ideal})               & $\ideal$          & $q^w$            & State only          & No; intractable proposal         \\
Semi-ideal (Alg.~\ref{alg:semiid})            & $\ir{n}$           & $\widetilde{q}^w_n$  & State only          & No; intractable acceptance         \\
Multiple-try  (Alg.~\ref{alg:mtm})                  & $\mtm{n}$         & $\widetilde{q}^w_n$ & State + auxiliary & Yes\\ \hline                  
\end{tabular}
\caption{Markov kernel notation used throughout this work. The proposal kernel $q^w$ is defined in \eqref{eq:qw} and $\widetilde{q}^w_n$ in \eqref{eq:qwn}.}
\label{table:kernels}
\end{table}

Given the empirical success of multiple-try approaches in various practical settings \cite{chang2022,frenkel1996,martino2018,liu2000,pandolfi2010}, it is natural to seek a theoretical justification for its good performance. In light of our derivation of Multiple-try, it would be natural to argue this point by way of comparison: ideally, $\ideal$ is a Markov kernel with excellent mixing properties, $\ir{n}$ is a good approximation of it, and $\mtm{n}$ is a good approximation of the latter. The goal of this work is to rigorously enact this heuristic argument to study the convergence behaviour of $\mtm{n}$.

Our first main result is a general $\mathrm{L}^2$ comparison result between $\ideal$ and $\mtm{n}$ in terms of their Dirichlet forms. This result depends on the moments of the associated \textit{importance weights} 
\begin{equation} \label{eq:impweights}
	\varpi(x,y):=\frac{\dif q^w(x,\cdot)}{\dif q(x,\cdot)}(y)=\frac{w(x,y)}{(qw)(x)},
\end{equation}
which can be interpreted as the weight normalized by the average value with respect to the proposal. Note that multiplying the weights $w$ by a $y$-independent factor yields an identical Markov chain; in this sense, $\varpi$ serves as a convenient standardisation of the weights. Let 
$$M_\varpi(p):=(\pi\otimes q) \left( \varpi(X,Y)^p \right)$$
denote their $p$-th moments.
\begin{theorem} \label{thm:mainwpi}
	For any $f\in\mathrm{L}^2(\pi)$ and $s>0$, there holds the inequality
	\begin{align*}
		\cal{E}(\ideal,f)&\leq s \cdot \cal{E}(\mtm{n},f) +  \beta_n(s) \cdot  \oscnorm{f}^2, \quad \forall s>0.  
	\end{align*}
	where $\beta_n(s)\to \mathbf{1}\{s\leq 1\}$ as $n \to \infty$. Furthermore, if $M_\varpi(2p)$ and $M_\varpi(-2p)$ are finite for some $p\in[1,\infty)$. 
	\begin{align*}
		&\beta_n(s) \leq  s^{-\frac{p^2}{1+2p}} \cdot \bigg\{ c_{2,n,p}\cdot\bigg(\frac{c_{1,n,p}p}{c_{2,n,p}(1+p)}\bigg)^{1+p}+c_{1,n,p}\cdot\bigg(\frac{c_{1,n,p}p}{c_{2,n,p}(1+p)}\bigg)^{-p} \bigg\}, \\
        &c_{1,n,p}:=M_\varpi(p)+K_1\cdot n^{-1}\quad \\ &c_{2,n,p}:= 2^{p+1}\cdot  \left\{K_2\cdot n^{-1}+K_3\cdot n^{-2}+  M_\varpi(2p) +M_\varpi(2p)\right\}. 
	\end{align*}
\end{theorem}
\begin{proof}
	This follows from Proposition \ref{pr:mainwpipre} and the bounds given in Propositions~\ref{pr:idealirwpi} and \ref{lemma:irmtmwpi}. The constants $K_1,K_2,K_3$ are explicit in the corresponding proofs.
\end{proof}
This comparison result is valid for any weight function, any $n$ and target measure $\pi$, provided the importance weights satisfy the moment conditions.
In particular, the comparison is tighter when the importance weights admit moments of larger order. We provide background on Dirichlet forms and Poincar{\'e} inequalities in Section \ref{sec:background} below. In particular, Theorem \ref{thm:mainwpi} allows us to  deduce convergence results for Multiple-try by studying the spectral gap of $\ideal$, which is a Metropolis--Hastings algorithm with proposal $q^w$. As we show, a common choice of weights $w(x,y)=\pi(y)/\pi(x)$ actually yields a Multiple-try algorithm whose spectral gap \textit{vanishes} with $n$, and we conclude that one should not use it in practice, strengthening a result of \cite{gagnon2023}. However, considering instead a choice of weight function recently introduced in \cite{gagnon2023,chang2022}, yields a limiting algorithm $\ideal$ that not only has a spectral gap, but has very good dimensional scaling properties.
For the following results, we specialise to the Gaussian setting, as has been done previously in theoretical analyses of Multiple-try \cite{gagnon2023}.

\begin{theorem}[Spectral gap bounds for $\ideal$] \label{thm:spgapbounds}
	Assume that $\pi(\dif x)=\cal{N}(\dif x;0,I_d)$, $w(x,y)=\sqrt{\pi(y)/\pi(x)}$ and that $q(x,\dif y)=\cal{N}(\dif y;x,\sigma^2 \cdot I_d)$. Then,
	\begin{align*}
		2^{-10}  \exp\left(-\frac{\sigma^4 \cdot d}{4}\right)  \sigma^2  (2+\sigma^2)  c_\gamma^2\leq \gamma(\ideal) \leq \frac{3}{2}\frac{\sigma^2}{2+\sigma^2} \wedge \bigg(1+\frac{\sigma^4}{(2+\sigma^2)^2}\bigg)^{-d/2}
	\end{align*}
	where $c_\gamma:=0.3177765$. Hence, among polynomial scalings, $\sigma \sim d^{-1/4}\Rightarrow \gamma(\ideal) \sim  d^{-1/2}$ is optimal.
\end{theorem}
\begin{proof}
	The lower and upper bound follow from Propositions \ref{prop:idealspgaplb} and \ref{prop:idealspgapub}. The optimality of $\sigma  \sim  d^{-1/4}$ follows from the following reasoning: if $\sigma  \sim  d^{-\beta}$ with $\beta\geq 1/4$, then $\gamma(\ideal) = \cal{O}(d^{-2\cdot \beta})$, and this rate is optimized when $\beta=1/4$. On the other hand, if $\beta<1/4$, then $\gamma(\ideal)$ goes to zero faster than $d^{-1/2}$, which is worse; hence we can conclude.
\end{proof}
If we take $\sigma = \zeta \cdot d^{-1/4}$ for some $\zeta>0$, in the bound above we obtain
\begin{equation*}
	2^{-10}\cdot \exp(-\zeta^4/4)\cdot \zeta^2 \cdot d^{-1/2}\cdot(2+\zeta^2 \cdot d^{-1/2})\cdot c_\gamma^2  \leq \gamma(\ideal) \leq (3/2)\cdot \zeta^2 \cdot d^{-1/2}.
\end{equation*}

\begin{remark}
In particular, the spectral gap of $\ideal$ has better scaling properties than Random Walk Metropolis, which is of order $d^{-1}$ with the optimal scaling of $\sigma \propto  d^{-1/2}$ \cite{alpw2022a} (and corresponds to the use of the `uninformed' weight $w \equiv 1 $). 

In the Random Walk Metropolis algorithm, its optimal diffusion scaling of $\sigma \propto  d^{-1/2}$ \cite{gelman1997}  also coincides with its spectral gap scaling aforementioned \cite{alpw2022a}. This does \textit{not} appear to be the case for $\ideal$, for which optimal diffusion scaling recommends scaling $\sigma$ to be of order $d^{-1/6}$ \cite{gagnon2023}. Since the spectral gap scaling is related to the \textit{minimum} acceptance probability across the space, it is expected that in general it might be worse than the kernel's diffusion scaling (which depends on the \textit{average} of the acceptance probability). 
\end{remark}

Further combining these results, we deduce the first explicit \textit{non-asymptotic} convergence bounds for Multiple-try Metropolis we are aware of, outside the context of independence samplers \cite{yang2023}. 
Note that in our bound below for technical reasons we consider the \textit{lazy} Multiple-try chain $P_{n, \ell}:=\frac{1}{2}(P_n+\mathsf{Id})$. This is a commonly used technique used to ensure positivity \cite{Lovasz1999,Vempala2005,Chen2020}, which is needed to deduce our bound. The price paid for this is a constant factor in our polynomial rate. (For the ideal chain, we explicitly show in Lemma~\ref{lemma:positive} that the chain is positive.)
\begin{theorem} \label{thm:mtmconvbound}
    Assume that $\pi(\dif x)=\cal{N}(\dif x;0,I_d)$, $w(x,y)=\sqrt{\pi(y)/\pi(x)}$ and that $q(x,\dif y)=\cal{N}(\dif y;x,\sigma^2 \cdot I_d)$. Then, if $\sigma^2<\sigma^2(p)$, where $\sigma^2(p)=\frac{1}{2p}-p+\frac{1}{2}(-3+\sqrt{5+\frac{1}{p^2}+\frac{2}{p}+12p+4p^2})$, 
    \begin{align*}
    	&\norm{\mtm{n,\ell}^k f }_{2,\pi}^2 
        \leq C_{n,\sigma,p}  \cdot k^{-\frac{p^2}{2(1+2p)}}
    \end{align*}
    for all $f\in\mathrm{L}_0^2(\pi)$ such that $\oscnorm{f}<\infty$. Here,
    \begin{align*}
        C_{n,\sigma,p}:=C_{\sigma}^{-\frac{(p+1)^2}{1+2p}} \cdot \bigg\{ c_{2,n,p}\cdot\bigg(\frac{c_{1,n,p}p}{c_{2,n,p}(1+p)}\bigg)^{1+p}+c_{1,n,p}\cdot\bigg(\frac{c_{1,n,p}p}{c_{2,n,p}(1+p)}\bigg)^{-p} \bigg\},
    \end{align*}
    $c_{1,n,p}$ and  $c_{2,n,p}$ are defined as in Theorem \ref{thm:mainwpi}, and $C_{\sigma}$ is the left hand side in the spectral gap estimate of Theorem \ref{thm:spgapbounds}.
\end{theorem}
\begin{proof}
	Combining Theorems \ref{thm:mainwpi} and \ref{thm:spgapbounds} via Lemma \ref{lemma:wpichaining} shows
	\begin{equation}
		\norm{f}_{2,\pi}^2 \leq s\cdot\cal{E}(P_n,f) + \beta_{n,\star}(s)\cdot \norm{f}_{\mathrm{osc}}^2 \leq 2s \cdot  \cal{E}(P_{n,\ell},f) + \beta_{n,\star}(s)\cdot \norm{f}_{\mathrm{osc}}^2 ,
        \label{eq:pf_1.3}
	\end{equation}
	with $\beta_{n,\star}(s)=\beta_n(C_{\sigma}s)/C_{\sigma}$, upon which if we rescale $s'=2s$ and use Lemmas \ref{lemma:wpisubg} and \ref{lemma:wpisubgres} to obtain the desired bound. The condition on $\sigma^2$ ensures the importance moments weights are finite for the given $p$: noting that $\varpi(x,y)\propto \sqrt{\pi(y)}\cdot e^{\frac{|x|^2}{2(2+\sigma^2)}}$, exponentiating and integrating the latter in $\pi\otimes q$, we can compute
	\begin{align*}
		M_\varpi(2p)&=\mathbb{E}\left[\omega(X,Y)^{2p}\right] = \bigg( \frac{2+\sigma^2}{2(1+p\sigma^2)} \bigg)^{dp/2}  \bigg(1+p\bigg(\frac{1}{1+p\sigma^2}-\frac{2}{2+\sigma^2}\bigg)\bigg)^{-d/2}
	\end{align*}
	and
	\begin{align*}
		M_\varpi(-2p)&=\mathbb{E}\left[\omega(X,Y)^{-2p}\right] \\
        &=  2^{dp/2}\cdot (2+\sigma^2)^{d(1-p)/2} \cdot \left(2+(1-3p-3p^2)\sigma^2-p\sigma^4\right)^{-d/2}
	\end{align*}
    provided the expressions in the square roots above are non-negative. The second is always lower than the first, and $\sigma^2<\sigma^2(p)$ implies the former is positive.
\end{proof}

The above bound holds for any $n\in\n$. From Theorem~\ref{thm:mainwpi}, we know that $\beta_n \to \textbf{1}_{[0,1]}$ as $n\to \infty$, and so for large $n$ we expect the convergence rate of Multiple-try Metropolis to approximate that of the ideal chain; see Section~\ref{sec:mtmcomparisonsideal}.

\subsection{Related work}
Despite Multiple-try being a well-known sampling technique, there are not many theoretical results available in the literature. An asymptotic scaling analysis of the algorithm is conducted in \cite{bedard2012}, in the case of high-dimensional product-form targets. \cite{yang2023} recently obtained the first convergence rate results we are aware of in the case of the independence sampler $q(x,\dif y)=q(\dif y)$, where one does not need to introduce balancing trials. \cite{gagnon2023,chang2022} introduce a novel weight function we consider in this work, and perform a scaling analysis and mixing time study, the latter in the case of discrete spaces. The analyses in \cite{liu2000,gagnon2023} are also based on the idea of using a limiting Multiple-try chain to draw conclusions about the Multiple-try algorithm, and, in a sense, we make those observations precise with the Poincar{\'e} inequalities introduced. 

We now comment on the independent and concurrent work of \cite{pozza2024}. Our two works are related, but are conducted with  different emphases and using different tools. In \cite{pozza2024}, the authors study a slightly more general class of algorithms (termed `multiproposal Markov chain Monte Carlo'), of which Multiple-try is the primary example, and they focus on delivering negative results, highlighting the fundamental computational limitations of such methods. As such, their results are couched in terms of \textit{upper bounds} on the spectral gaps of these chains. By contrast, while we do also derive some negative results (such as Proposition~\ref{prop:spidifferentn}), the focus of our work is to provide \textit{positive} results and convergence guarantees via Markov chain comparisons. Taken together, these two papers can be seen as a complementary pair of works, with \cite{pozza2024} supplying fundamental limitations on the best-case performance of Multiple-try type chains, and our present work supplying guarantees for the worst-case situation. We will comment in more depth on the similarities and differences of our works throughout the paper.

\subsection{Comparison of Markov chains} \label{sec:background}

We introduce some key notation, and then some results and definitions from the framework of weak Poincar\'e inequalities from \cite{alpw2022b,alpw2022c} we will use throughout. For a function $f:\Xspace\to \r$ and a measure $\pi$ on $(\Xspace,\mathcal{F})$, we denote $\pi(f):=\int \pi(\dif x) f(x)$. We let $\mathrm{L}^2(\pi)$ denote the  set of square integrable real-valued functions $\pi(f^2)<\infty$, endowed with the inner product $\iprod{f}{g }_{2,\pi}:=\pi(f\cdot g)$ and induced norm $\norm{f}^2_{2,\pi}:=\iprod{f}{f }_{2,\pi}$. We further introduce the space $\mathrm{L}_0^2(\pi):=\{ f\in \mathrm{L}^2(\pi): \pi(f)=0\}$ of the mean zero elements of $\mathrm{L}^2(\pi)$ and the oscillation seminorm $\oscnorm{f}:=\ess_\pi \sup f - \ess_\pi \inf f$. For two probability measures $\mu$ and $\nu$ on $(\Xspace,\mathcal{F})$, we denote $(\mu\otimes\nu)(A\times B)=\mu(A)\nu(B)$ for $A,B\in\mathcal{F}$. $\mathsf{Id}$ denotes the \textit{identity operator} on $\mathrm{L}^2(\pi)$. Given a $\pi$-invariant Markov kernel $P$, the Dirichlet form associated to the pair $(P,\pi)$ is defined for any $f\in\mathrm{L}_0^2(\pi)$ by $\cal{E}(P,f):=\iprod{(\mathsf{Id}-P)f}{f }_{2,\pi}$. We will also often use the alternative representation
\begin{equation*}
	\cal{E}(P,f)=\frac{1}{2}\int \pi(\dif x)P(x,\dif y)(f(x)-f(y))^2.
\end{equation*}
$P$ is said to be \textit{positive} if $P$ is reversible and $\iprod{Pf}{f }_{2,\pi}\geq 0$ for all $f\in\mathrm{L}^2(\pi)$.
\begin{definition}[Standard Poincar{\'e} inequality; SPI]
    We say that a $\pi$-reversible positive Markov kernel $P$ satisfies a Standard Poincar{\'e} inequality (SPI) with constant $c_P>0$ if for all $f\in \mathrm{L}_0^2(\pi)$,
    \begin{equation} \label{eq:SPI}
    	c_P\norm{f }_{2,\pi}^2\leq \cal{E}(P,f).
	\end{equation}
\end{definition}
By iterating \eqref{eq:SPI}, we can immediately characterize the $\mathrm{L}_0^2(\pi)$-exponential convergence of $P$ in the following sense.
\begin{lemma}[SPI $\Rightarrow$ geometric convergence]  \label{lem:SPI}
	Let $P$ be a $\pi$-reversible positive Markov kernel satisfying a SPI. Then, for all $f\in\mathrm{L}_0^2(\pi)$ we have for any $n\in\n$ that
    \begin{equation*}
	   \norm{P^k f }_{2,\pi}^2 \leq (1-c_P)^k\norm{f }_{2,\pi}^2.
	\end{equation*}
\end{lemma}
\begin{proof}
    Since $P$ is positive, we have that $\cal{E}(P,f)\leq \cal{E}(P^* P,f)$ by \cite[Lemma~18]{alpw2022b}. 
    It then follows that 
    \begin{equation*}
        c_P\norm{f }_{2,\pi}^2\leq \cal{E}(P,f)  \leq\cal{E}(P^* P,f)= \norm{f }_{2,\pi}^2 - \norm{Pf }_{2,\pi}^2 \hspace{0.1em}\Rightarrow \hspace{0.1em}\norm{Pf }_{2,\pi}^2 \leq (1-c_P)\norm{f }_{2,\pi}^2
    \end{equation*}
for all $f\in\mathrm{L}_0^2(\pi)$, and  the claim follows upon iterating the last inequality $n$ times.
\end{proof}

In various cases of practical interest, one encounters Markov kernels which are not exponentially ergodic, and so cannot satisfy any SPI. Nevertheless, for these kernels which converge only at slower-than-exponential (`subgeometric') rates, we can still obtain fine control on their convergence behaviour in various ways. The following class of functional inequalities is helpful in characterizing these.
\begin{definition}[Weak Poincar{\'e} inequality; WPI]
    We say that a $\pi$-reversible positive Markov kernel $P$ satisfies a weak Poincar{\'e} inequality (WPI) if for all $f\in \mathrm{L}_0^2(\pi)$,
    \begin{equation} \label{eq:WPI}
	   \norm{f }_{2,\pi}^2 \leq s \cdot \cal{E}(P,f) + \beta(s)\cdot \oscnorm{f}^2, \quad \forall s>0,
	\end{equation}
    where $\beta:(0,\infty)\to [0,\infty)$ is a decreasing function such that $\lim_{s\rightarrow \infty}\beta(s)=0$.
\end{definition}

\begin{lemma}[WPI $\Rightarrow$ sub-geometric convergence]  \label{lemma:wpisubg}
    Let $P$ be a $\pi$-reversible positive Markov kernel satisfying a WPI. Then, for all $f\in \mathrm{L}_0^2(\pi)$ such that $\oscnorm{f}<\infty$ we have for any $k\in\n$ that
    \begin{align} \label{eq:wpisubg}
	   \norm{P^k f}^2_{2,\pi} \leq F^{-1}(k) \cdot \oscnorm{f}^2
	\end{align}
    where $F$ is a convex decreasing invertible function defined by
    \begin{align*}
	   F(k):=\int_k^1 \frac{1}{K^*(v)}\dif v, \quad \text{where}
	\end{align*} 
    \begin{equation*}
        K^*(v):=\sup_{u\geq 0}\{uv-K(u)\} \quad \text{and} \quad K(u):=
		\begin{cases*}
            u\cdot \beta(1/u), & \quad \text{if}\quad u\textgreater0 \\ 0, & \quad \text{if} \quad u=0.
		\end{cases*}
	\end{equation*}
\end{lemma}
\begin{example}
    If \eqref{eq:WPI} holds with $\beta(s)=cs^{-\alpha}$, then $F^{-1}(n)\leq c (1+\alpha)^{1+\alpha} n^{-\alpha}$.
\end{example}
\begin{lemma}[Linear rescaling] \label{lemma:wpisubgres}
    Under the same conditions of Lemma \ref{lemma:wpisubg}, if, instead we have the rescaled WPI
    \begin{equation*} 
	   \norm{f }_{2,\pi}^2 \leq s \cdot \cal{E}(P,f) + \beta(cs)\cdot \oscnorm{f}^2, \quad \forall s>0,
	\end{equation*}   
    the bound \eqref{eq:wpisubg} holds with $c\cdot n$ in place of $n$.
\end{lemma}

SPIs and WPIs can also be used to deduce convergence properties of a given Markov kernel \textit{relative to} another one; see \cite[Theorem~33]{alpw2022b}.
\begin{lemma}[Chaining]\label{lemma:wpichaining} 
    Let $P_1$, $P_2$ and $P_3$ three $\pi$-invariant Markov kernels. Assume that for all $s>0$ and $f\in\mathrm{L}^2_0(\pi)$,
    \begin{align*}
	   \cal{E}(P_1,f) &\leq s\cdot \cal{E}(P_2,f) + \beta_1(s)\oscnorm{f}^2, \\
	   \cal{E}(P_2,f) &\leq s\cdot\cal{E}(P_3,f) + \beta_2(s)\oscnorm{f}^2. 
	\end{align*}
Then, 
\begin{equation*}
    	   \cal{E}(P_1,f) \leq s\cdot\cal{E}(P_3,f) + \beta(s)\oscnorm{f}^2, \quad \forall s>0, f\in\mathrm{L}^2_0(\pi),
\end{equation*}
    where $\beta(s)=\inf \{\beta_1(s_1)+s_1\beta_2(s_2)\mid s_1>0,s_2>0, s_1s_2=s \}$. In particular, if $\beta_i(s) = c_i s^{-p}$ for $i\in\{1,2\}$, then $\beta(s)=cs^{-\alpha}$ with $\alpha:=\frac{p^2}{1+2p}$ and $c:=c_2\cdot(\frac{c_1p}{c_2(1+p)})^{1+p}+c_1\cdot(\frac{c_1p}{c_2(1+p)})^{-p}$.
\end{lemma} 
If $P_1(x,\dif y)$ is the `perfect' kernel $\pi(\dif y)$, the above result also illustrates how we can deduce convergence bounds for $P_3$ given a WPI between $P_2$ and $P_3$.
\begin{remark}\label{rmk:spi_beta}
	If $\beta_1(s)=0$ for all $s\geq c_{P_1}^{-1}$ for some $c_{P_1}>0$, then a WPI implies a SPI between $P_2$ and $P_1$ with constant $c_{P_1}$. In particular, since any SPI constant satisfies $c_P\le 1$, a $\beta$-function of $\beta_1(s)=\mathbf{1}\{s\le 1\}$ is in this sense the best possible. In this case, the convergence rate of $P_2$ is arbitrarily close to that of $P_1$. If in particular $P_1(x,\dif y)=\pi(\dif y)$, this means that $P_2$ has a spectral gap equal to 1.
\end{remark} 

The following is a practical criterion to establish WPIs and SPIs we will often use; see \cite[Theorem~36]{alpw2022b}.
\begin{lemma}[Off-diagonal ordering] \label{lemma:wpifromdomination} 
    Let $P_1$ and $P_2$ two $\pi$-invariant Markov kernels. Assume that for any $(x,A)\in\Xspace\times\mathcal{F}$, $P_2(x,A \backslash \{x\}) \geq \int_{A \backslash \{x\}} \eta(x,y)P_1(x,\dif y)$ for some $\eta:\Xspace^2\to (0,\infty)$. Then, for any $s>0$ and any $f\in\mathrm{L}_0^2(\pi)$ such that $\oscnorm{f}<\infty$,
    \begin{equation*}
	   \cal{E}(P_1,f) \leq s\cdot\cal{E}(P_2,f) + \beta(s)\cdot \oscnorm{f}^2
	\end{equation*}
    where we defined 
    \begin{equation*}
    	\beta(s):= \frac{1}{2}(\pi\otimes P_1)(A(s)^\complement\cap \{X\neq Y\}), \qquad A(s):=\{(x,y)\in\Xspace^2:\eta(x,y)>1/s\}.
    \end{equation*}
    In particular, if $\eta(x,y)\geq c_{P_2}^{-1}$ for some $c_{P_2}>0$ and all $(x,y)\in\Xspace^2$, then $\cal{E}(P_1,f) \leq c_{P_2}\cdot\cal{E}(P_2,f)$. 
\end{lemma}
\section{Comparison of Multiple-try chains} \label{sec:comparisons}
In this section we derive various comparison results for Multiple-try chains. For the reader's convenience, we have listed the various kernels and notations in Appendix~\ref{sec:usefulexpressions}.
\subsection{Optimal number of proposals} \label{sec:optn}

Using an appropriate SPI, we first establish a \textit{negative result} for Multiple-try Metropolis: increasing the number of proposals can only achieve a limited improvement.
\begin{proposition} \label{prop:spidifferentn}
	For all $n\in\n$ and $f\in\mathrm{L}^2(\pi)$, it holds
	\begin{equation*}
		(n-1)\cdot\cal{E}(\ir{n},f) \leq n\cdot\cal{E}(\ir{n-1},f), \quad \textup{and} \quad (n-1)\cdot\cal{E}(\mtm{n},f) \leq n\cdot\cal{E}(\mtm{n-1},f).
	\end{equation*}
\end{proposition}
\begin{proof}
	See Appendix~\ref{subsec:pf_MTM_neg}.
\end{proof}
Proposition \ref{prop:spidifferentn} establishes that incrementing by 1 the number of proposals in Multiple-try Metropolis can only improve performance by \textit{at most} $n/(n-1)$ in the sense of the Dirichlet forms. Thus, the maximal improvement of the spectral gap is bounded as
\begin{equation*}
	(n-1)\cdot\gamma(\ir{n}) \leq n\cdot\gamma(\ir{n-1}), \quad \textup{and} \quad (n-1)\cdot\gamma(\mtm{n}) \leq n\cdot\gamma(\mtm{n-1}).
\end{equation*}
Notably, for example, if the Multiple-try chain is \textit{sub-geometric} (so $\gamma(\mtm{n})=0$), no increase in the number of proposals can ever make the chain geometric. Interestingly, one can have chains which have a positive spectral gap at any finite $n$, but which vanishes in the large $n$ limit.
In light of this result, we can argue that for Multiple-try with conditionally i.i.d.~proposals, with a \textit{serial} implementation, in terms of the efficiency per proposal, $n = 1$, i.e. the simple Metropolis--Hastings scheme, is optimal. Hence, to obtain better practical performance from Multiple-try schemes, one needs either non-i.i.d.~proposals, or a non-serial implementation, motivating approaches such as \cite{casarin2013,craiu2007}. We notice that a similar optimality of a single pseudo-simple has also been observed in \cite{bornn2017}, but in a different context involving ABC algorithms; as far as we can tell, neither result implies the other.

We show in Proposition \ref{prop:spectralgapgbmtm} that for a common implementation of Multiple-try, it is even true that the spectral gap \textit{vanishes} with $n$. This can be interpreted as follows: the importance resampling approximation is effectively pushing the effective proposal close to the `ideal' $q^w$, which happens to be a very poor proposal, for which the corresponding ideal Metropolis--Hastings algorithm possess no spectral gap. We discuss this more in detail later.
\begin{remark}
	\cite{yang2023} arrives at an analogous conclusions in the restricted case of Multiple-try with state-independent proposals. The concurrent work \cite{pozza2024} presents a very similar result by comparing the spectral gap of $P_n$ to a single-proposal chain and an analogous conclusion on the optimality of $n=1$ in the serial context. Proposition \ref{prop:spidifferentn} also isolates that this behaviour is caused by the importance resampling approximation step, and it is slightly more general in that provides comparison between each pair $(P_{n},P_{m})$ via induction.
\end{remark}
\subsection{Comparison to ideal algorithms} \label{sec:mtmcomparisonsideal}

In this section we derive comparison results between Multiple-try $\mtm{n}$ and ideal chains $\ideal$. When the importance weights admit finite moments of larger order, the comparison is tighter. These results also allow us to deduce convergence bounds for $\mtm{n}$ by analysing $\ideal$, which is done in the next section.

\paragraph{Comparison between ideal and semi-ideal kernels.}  This comparison quantifies the loss in performance caused by the importance resampling approximation to $q^w$. This loss can be characterized by studying the fluctuations of the importance weights \eqref{eq:impweights}. 
Whenever these can be uniformly bounded above, it is possible to establish a SPI between $\ir{n}$ and $\ideal$.

\begin{lemma} \label{lemma:idealirspi}
	Assume that $|\varpi|_{\infty}<\infty$. Then, 
	\begin{equation*}
		\cal{E}(\ir{n},f)\geq |\varpi|_{\infty}^{-1} \cdot \cal{E}(\ideal,f) \quad \forall f\in\mathrm{L}^2(\pi).
	\end{equation*}
\end{lemma}
\begin{proof}
	See Appendix \ref{app:idealirspi}.
\end{proof}
Such uniform bounds are expected to be difficult to ensure outside of compact state spaces. When the importance weights are not uniformly bounded, but their moments are at least finite, we can still relate their convergence properties via a WPI. To this end, recall that
\begin{equation*}
	M_\varpi (p) := (\pi\otimes q)(\varpi(X,Y)^p).
\end{equation*}
denotes the $p$-th moment of the importance weights.
\begin{proposition}\label{pr:idealirwpi}
	For any $f\in\mathrm{L}^2(\pi)$ and $s>0$,
    \begin{align*}
		&\cal{E}(\ideal,f)
		\leq s\cdot \cal{E}(\ir{n},f) +  \oscnorm{f}^2\cdot \beta_{1,n}(s), \\ &\beta_{1,n}(s):= (\pi\otimes q^w)\bigg(\mathbb{E}\bigg[\frac{n}{\sum_{i=1}^n \varpi(X,Y_i)}\mid Y_1=Y\bigg] < \frac{1}{s} \bigg).
	\end{align*}
    Furthermore, if in addition $M_\varpi(p+1)$ is finite for some $p\in[1,\infty)$, then, for all $s>0$,
	\begin{align}
		\beta_{1,n}(s)\leq s^{-p} \cdot \left\{ K_1 \cdot n^{-1}+ M_\varpi (p) \right\}. \nonumber
	\end{align}
\end{proposition}
\begin{proof} 
	See Appendix \ref{app:idealirwpi}. The constant $K_1$ can be traced in the proof.
\end{proof}
\begin{remark}\label{rmk:beta_conv1}
    We see from the dominated convergence theorem and the strong law of large numbers that as $n\to \infty$, for each $s>0$, $\beta_{1,n}(s)\to  \mathbf{1}\{s\le 1\}$, corresponding to the best possible comparison between $\ideal$ and $\ir{n}$, see Remark~\ref{rmk:spi_beta}. 
\end{remark}

\paragraph{Comparison between the semi-ideal and the Multiple-try Metropolis kernels.} 

This comparison captures the loss in performance due to the introduction of the shadow proposals to make the acceptance probability of $\ir{n}$ computable. Because these also appear in the denominator of the acceptance function, we need some control on their distance to zero. In particular, if these are also uniformly lower bounded, we immediately obtain a strong comparison result.
\begin{lemma} \label{lemma:irmtmspi}
	Assume that $|\varpi|_\infty<\infty$ and $|\varpi^{-1}|_\infty<\infty$. Then,
	\begin{equation*}
		\cal{E}(\mtm{n},f) \leq  |\varpi^{-1}|_\infty^2 \cdot |\varpi|_\infty^{2} \cdot \cal{E}(\ir{n},f)
	\end{equation*}
\end{lemma}
\begin{proof}
	See Appendix \ref{app:irmtmspi}.
\end{proof}
These uniform lower boundedness conditions can be relaxed via a WPI approach via an appropriate control on the negative moments of the importance weights.
\begin{proposition} \label{lemma:irmtmwpi}
    For any $f\in\mathrm{L}^2(\pi)$ and $s>0$,
    \begin{align*}
        &\cal{E}(\ir{n},f) 
		\leq s\cdot \cal{E}(\mtm{n},f) + \oscnorm{f}^2\cdot \beta_{2,n}(s), \\ 
        &\beta_{2,n}(s):=(\pi\otimes \tilde{P}_n)(\zeta_n(X,Y)<1/s), \\
        &\zeta_n(x,y):=\Ebb{ \min\left(\frac{(\widetilde{qw})_n(x,y)}{(\widehat{qw}_n)(x,Y^{[n]})}, \frac{(\widetilde{qw})_n(y,x)}{(\widehat{qw}_n)(y,Z^{[n]})} \right) \mid  Y_1=y,Z_1=x }
    \end{align*}
	Furthermore, if in addition $M_\varpi(2p)$ and $M_\varpi(-2p)$ are finite for some $p\in[1,\infty)$, then, for all $s>0$,
	\begin{align*}
		\beta_{2,n}(s) \leq s^{-p}\cdot 2^{p+1}\cdot  \left\{K_2 \cdot n^{-1}+K_3\cdot n^{-2}+  M_\varpi(2p) +M_\varpi(-2p)\right\}. 
	\end{align*}
\end{proposition}
\begin{proof}
	See Appendix \ref{app:irmtmwpi}. The constants $K_2,K_3$ can be traced in the proof.
\end{proof}
\begin{remark}
    Similar to Remark~\ref{rmk:beta_conv1}, we have that $\beta_{2,n}(s)\to \mathbf{1}\{s\le 1\}$ as $n\to \infty$, which is again the best possible comparison result.
\end{remark}

\paragraph{Comparison between Multiple-try and the ideal kernel.} 

Chaining Proposition \ref{pr:idealirwpi} and \ref{lemma:irmtmwpi} via Lemma \ref{lemma:wpichaining} returns the following comparison result between Multiple-try and the ideal Metropolis kernel.
\begin{proposition} \label{pr:mainwpipre}
For any $f\in\mathrm{L}^2(\pi)$ and any $s>0$, there holds
\begin{align*}
    &\cal{E}(\ideal,f)\leq s \cdot \cal{E}(\mtm{n},f) + \oscnorm{f}^2\cdot \beta_n(s), \\
    &\beta_n(s)=\inf\{\beta_{1,n}(s_1)+s_1\beta_{2,n}(s_2) | s_1>0, s_2>0, s_1s_2 = s\}
\end{align*}
\end{proposition}
In particular,  as $n\rightarrow \infty$, we have $\beta_n(s)\to \mathbf{1}\{s\le 1\}$. Therefore, in this sense, in the limit of multiple trials, the convergence rate of Multiple-try is `close' to that of the ideal kernel. In fact, with an argument somewhat similar to that in \cite[pp 25-26]{alpw2022b}, if $\ideal$ satisfies a SPI with constant $c_{\ideal}$, again Lemma \ref{lemma:wpichaining} shows
\begin{equation*}
   \norm{f }_{2,\pi}^2 \leq s\cdot \cal{E}(P_n,f) + \oscnorm{f}^2\cdot \bigg\{\frac{s}{1+\epsilon}\beta_n(1+\epsilon)+\mathbf{1}\{ s\leq (1+\epsilon)c_{\ideal}^{-1}\}\bigg\}
\end{equation*}
for any $\epsilon>0$, in particular,
\begin{equation*}
    \frac{s}{1+\epsilon}\beta_n(1+\epsilon)+\mathbf{1}\{ s\leq (1+\epsilon)c_{\ideal}^{-1}\} \to \mathbf{1}\{ s\leq (1+\epsilon)c_{\ideal}^{-1}\}
\end{equation*}
so that for large $n$, Multiple-try can be seen as an approximation of $\ideal$ (see Remark \ref{rmk:spi_beta}). This holds for any choice of weight function, any target $\pi$ and any choice of proposal $q$. When the importance weights possess suitable moments, one can bound $\beta_n(s)$ via the corresponding bounds for $\beta_{1,n}(s)$ and $\beta_{2,n}(s)$ in Propositions \ref{pr:idealirwpi}, \ref{lemma:irmtmwpi} and prove Theorems \ref{thm:mainwpi} and \ref{thm:mtmconvbound}.
\begin{remark} \label{remark:gagnonweakconv}
	Theorem 1 in \cite{gagnon2023} assumes fourth moment conditions on the weights $w$ to prove weak convergence of the Multiple-try Metropolis Markov chain to the ideal chain when starting from stationarity. We argued that Multiple-try achieves a convergence rate close to that of the ideal chain as $n\rightarrow \infty$ under general conditions. Under moment conditions on the importance weights, Theorem \ref{thm:mainwpi} gives a non-asymptotic comparison result. The moment conditions in \cite{gagnon2023} can be shown to imply ours on the importance weights via a Cauchy--Schwarz inequality argument.
\end{remark} 
\section{Convergence analysis of the ideal Markov chain} \label{sec:idealanalysis}
In this section we analyse the convergence properties of $\ideal$ (and thus of $\mtm{n}$ via Theorem \ref{thm:mainwpi}). We consider the case when $q(x,\dif y)=\cal{N}(\dif y;x,I_d\cdot\sigma^2)$ is a Gaussian Random Walk, which is a common choice in the implementation of Multiple-try schemes. We consider two choices of weight function.

(i) The \textit{globally balanced} weights $w(x,y)=\pi(y)/\pi(x)$. This choice is commonly used in Multiple-try algorithms used in practice \cite{martino2017,martino2018,gagnon2023}. The rationale behind this choice is that one should give more weight to points that are in high-probability areas of $\pi$. However, as we illustrate in Section \ref{sec:gbanalysis},  strengthening some results in \cite{gagnon2023}, this choice the algorithm is excessively biased towards high-probability regions of $\pi$, resulting in a Markov chain whose spectral gap that vanishes as the number of proposals grow (Proposition \ref{prop:spectralgapgbmtm}).

(ii) The \textit{locally balanced} weights $w(x,y)=\sqrt{\pi(y)/\pi(x)}$ were recently introduced in \cite{gagnon2023,chang2022}.
The rationale behind this choice, inspired by the work of \cite{zanella2020}, is that the resulting proposal kernel $q^w$ becomes $\pi$-invariant as $\sigma\rightarrow 0$, resulting in a Markov chain with higher expected acceptance probability in the low $\sigma$ scenario. In higher dimensions, where one usually has to take near-zero scale parameter $\sigma$, this choice is thus expected to return a better mixing Markov chain. With globally balanced weights, $q^w$ leaves $\pi$ invariant in the case $\sigma\rightarrow \infty$, while in the small $\sigma$ regime it is only invariant for $\pi^2$, which might look quite different from $\pi$ in high-dimensions \cite{zanella2020,gagnon2023}. The better scaling of the locally balanced choice is shown in \cite{gagnon2023} via an optimal scaling analysis, and we provide further support to this claim by analyzing the scaling of the spectral gap of $\ideal$ with the dimension.
We develop our analysis for Gaussian targets, similarly to \cite{gagnon2023}. While this is certainly restrictive, we note that if $\pi$ is a Bayesian posterior distribution, then under mild regularity conditions, the Bernstein-von Mises theorem implies that in a suitably data-rich limit, $\pi$ can be expected to admit suitably Gaussian-like behaviour, providing some heuristic support for beginning our exploration in this rather idealised setting.
\begin{assumption} \label{hp:pigaussian}
    $\pi$ has a standard Gaussian density: $\pi(\dif x)=\cal{N}(\dif x;0,I_{d})$.
\end{assumption}
Under the Gaussian assumption on the target, for a generic weight function of the form $w(x,y)=(\pi(y)/\pi(x))^\theta$ for $\theta\ge 0$ (which includes the cases above), the proposal kernel $q^w$ defines Gaussian transitions:
\begin{equation} \label{eq:qwisgaussian}
	q^w(x,\dif y) = \cal{N}\bigg(\dif y;x\cdot \frac{1}{1+\theta \sigma^2};I_d \cdot\frac{\sigma^2}{1+\theta \sigma^2}\bigg).
\end{equation}
The resulting ideal Metropolis kernel $\ideal$ is always positive:
\begin{lemma} \label{lemma:positive}
    $\ideal$ is a positive Markov kernel: it is reversible and $\iprod{\ideal f}{f}_{2,\pi}\geq 0$ for all $f\in\mathrm{L}^2(\pi)$.
\end{lemma}
\begin{proof}
    Reversibility is clear since it is a Metropolis--Hastings chain. Let $\rho:=1/(1+\theta\sigma^2)$, $\gamma:=\sigma^2/(1+\theta\sigma^2)$. We can check that $q^w$ is reversible with respect to $\nu(\dif x)=\cal{N}(\dif x;0,I_d\cdot\gamma/(1-\rho^2))$, and the result follows by \cite[Proposition 3]{Doucet2015} upon identifying, in their notation, $\chi=\nu$ and $r(u,v)=\cal{N}(u;\rho^{1/2}v,\gamma(1-\rho))$.
\end{proof}
\subsection{Globally balanced weights} \label{sec:gbanalysis}
Here, we use a conductance approach to show that the spectral gap of $\mtm{n}$ with globally balanced weights vanishes as the number of proposals $n$ increases.
\begin{definition}[Conductance]
    The conductance of a $\pi$-invariant Markov kernel $P$ is 
    \begin{equation} \label{eq:conductance}
	   \Phi(P):=\inf \bigg\{ \frac{(\pi\otimes P)(A\times A^c)}{\pi(A)}:A\in \cal{F},\pi(A)\leq 1/2 \bigg\}.
	\end{equation}
\end{definition}
\begin{lemma}[Cheeger's inequalities]\label{lemma:cheeger} 
    For a $\pi$-reversible positive Markov kernel $P$, it holds that
	\begin{equation*}
	   2^{-1}\Phi(P)^2 \leq \gamma(P) \leq 	\Phi(P)
	\end{equation*}
	where $\gamma(P)$ denotes the spectral gap of $P$.
\end{lemma}
\begin{proposition} \label{prop:spectralgapgbmtm}
    Let Assumption \ref{hp:pigaussian} hold. Then, as $n\rightarrow \infty$, $ \Phi(\mtm{n})\rightarrow 0$.
\end{proposition}
\begin{proof}
    See Appendix \ref{app:spectralgapgbmtm}.
\end{proof}
By Cheeger's inequality, this also implies that the spectral gap of $P_n$ vanishes. In light of this result, it is hard to recommend using Multiple-try schemes with weight function  $w(x,y)=\pi(y)/\pi(x)$ with a random walk proposal. This phenomena is inherently related to the fact that the globally balanced Multiple-try is excessively biased towards high probability regions of $\pi$, and can be understood by considering the example presented in \cite{martino2017}, illustrated in Figure~\ref{fig:MTM_props}. 
\begin{figure}[t]
\includegraphics[width=1\textwidth]{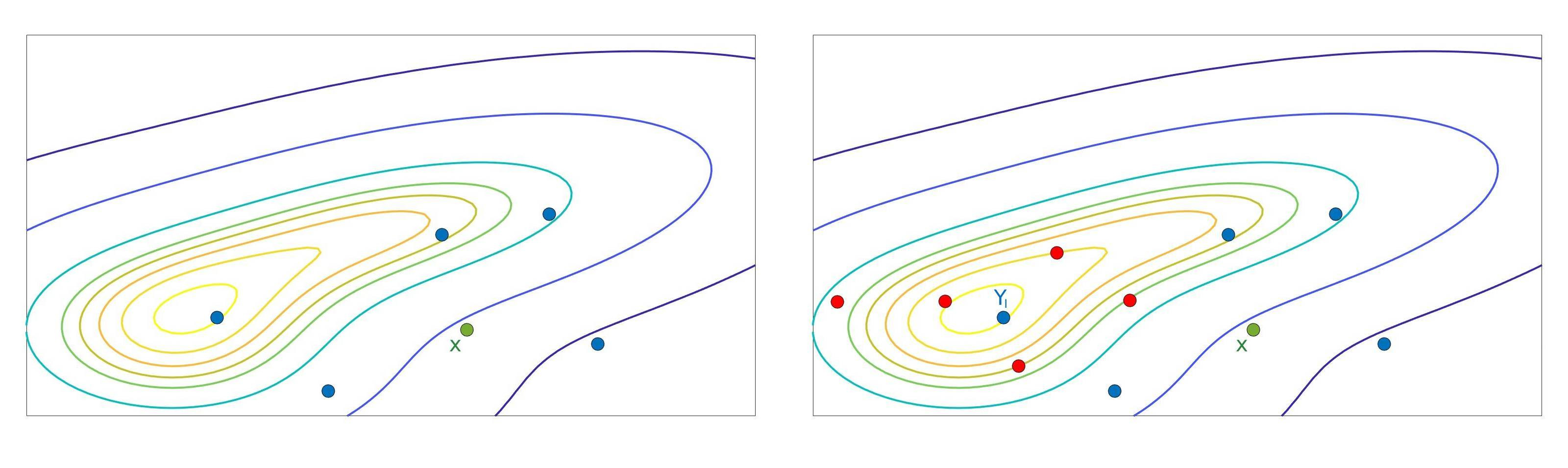}
 \caption{Degenerate behaviour of globally balanced Multiple-try Metropolis: the shadow samples are more likely to be drawn in higher probability areas, resulting in slow convergence.}  \label{fig:MTM_props}
\end{figure}
Figure~\ref{fig:MTM_props} can be interpreted as follows. Suppose the current sample $x$ is the green circle, which lies in a lower probability region. Among the proposals in blue, some will also lie in low probability regions, and others in higher regions. This is illustrated in the left plot. With globally balanced weights, one of the high probability proposals is most likely to be chosen as effective proposal, which we marked as $Y_I$. However, this leads to all the shadow samples, marked in red, to be also located in higher probability regions. Consequently, the acceptance probability, $\min\{1,\sum_{i=1}^n \pi(Y_i)/\sum_{i=1}^n \pi(Z_i)\}$, is likely to be low. Notice that this phenomena can only worsen with $n$. 
On the other hand, this bad behavior cannot occur if we consider an independence sampler $q(x,\dif y)=q(\dif y)$, as each proposal sample is independent of the current position. \cite{yang2023} shows that the independence globally balanced Multiple-try Metropolis can in fact be uniformly geometrically ergodic, and that the spectral gap does not vanish with $n$.

\subsection{Locally balanced weights}
We turn the analysis to the locally balanced $\ideal$. We show that with this choice of weight function $\ideal$ is geometrically convergent, has better scaling properties with the dimension than Random Walk Metropolis, and we derive lower and upper bounds on its spectral gap that reflects this -- see the discussion in Section \ref{sec:mtmintro}.
Firstly, we notice that under the Gaussian assumption on the target, the proposal $q^w$ is Gaussian, given by \eqref{eq:qwisgaussian} with $\theta=1/2$, and we can write, 
\begin{align} \label{eq:alphainf}
    \alphid(x) &= \int \cal{N}\bigg(\dif y;x\cdot \frac{2}{2+\sigma^2};I_d \cdot\frac{2 \sigma^2}{2+\sigma^2} \bigg) \alpha_\infty(x,y), \\ \alpha_\infty(x,y)&=\min\left\{1,\frac{\pi(\dif y)q^w(y,\dif x)}{\pi(\dif x)q^w(x,\dif y)}\right\}
\end{align}
and we can calculate explicitly $\alpha_\infty(x,y)$ as
\begin{equation*}
	\alpha_\infty(x,y) = \min\left\{ 1, \exp(-\psi(y)+\psi(x))\right\}, \quad \text{with} \quad \psi(x)=\frac{\sigma^2}{4(2+\sigma^2)}\norm{x}^2 + \mathrm{const}.
\end{equation*} 
This facilitates the analysis, by allowing us to extend the techniques introduced in \cite{alpw2022a}. In this section we prove Theorem \ref{thm:spgapbounds} by proving upper and lower bounds to the spectral gap of $\ideal$.

\subsubsection{Spectral gap lower bound}

The main result of this section is the following.
\begin{proposition} \label{prop:idealspgaplb} 
    Assume Assumption \ref{hp:pigaussian}. Let $\sigma=\zeta\cdot d^{-1/4}$. For any $\zeta>0$,
    \begin{align*}
	   \Phi(\ideal) &\geq 2^{-9/2}\cdot \exp(-\zeta^4/8)\cdot \sqrt{\zeta^2 \cdot d^{-1/2}\cdot(2+\zeta^2\cdot d^{-1/2})} \cdot c_\gamma \quad \text{and} \\
	   \gamma(\ideal) &\geq 2^{-10}\cdot \exp(-\zeta^4/4)\cdot \zeta^2 \cdot d^{-1/2}\cdot(2+\zeta^2\cdot d^{-1/2}) \cdot c_\gamma^2. 
	\end{align*}
\end{proposition}
\begin{proof}
    This follows from combining Proposition \ref{prop:alpwthm18}, Corollary \ref{cor:pidealclosecoulping} and Lemma \ref{lemma:alphidlb}, and noting that $2^2 \cdot \delta^2\cdot  c_\gamma^2 \ge \exp\left(-\frac{\zeta^4}{8}\right)\cdot \zeta^2 \cdot d^{-1/2}\cdot(2+\zeta^2 \cdot d^{-1/2})\cdot 2^{-1} \cdot c_\gamma^2 $ and that
    \begin{equation*}
        \exp\bigg(-\frac{\zeta^4}{8}\bigg)\cdot \zeta^2 \cdot d^{-1/2}\cdot(2+\zeta^2 \cdot d^{-1/2})\cdot 2^{-1} \cdot c_\gamma^2 \leq 4\cdot c_\gamma^2 \cdot ( e^{-1/2} + 1 )  < 1
    \end{equation*}
\end{proof}
The proof of this result requires the following \textit{close coupling} condition.
\begin{definition}
    For given $\epsilon, \delta >0$, we say that a Markov kernel $P$ on a metric space $(\Xspace,\mathsf{d})$ is $(\mathsf{d},\delta,\varepsilon)$-close coupling if
    \begin{equation*}
	   \mathsf{d}(x,y) \leq \delta \Rightarrow \tvdist{P(x,\cdot)-P(y,\cdot)} \leq 1-\varepsilon, \quad x,y\in\Xspace.
	\end{equation*}
\end{definition}
\begin{proposition}[Theorem 18 and Lemma 27 in \cite{alpw2022a}] \label{prop:alpwthm18}
	Assume Assumption \ref{hp:pigaussian} and let $\ideal$ be a $(\mathsf{d},\delta,\varepsilon)$-close coupling, $\pi$-reversible and positive Markov kernel. Then, it holds that
    \begin{equation*}
	   \Phi(\ideal) \geq 2^{-2}\cdot \varepsilon\cdot \min\{1,2\cdot \delta \cdot c_\gamma\} \quad \text{and} \quad \gamma(\ideal)\geq 2^{-5}\cdot \varepsilon^2\cdot \min\{1,2^2\cdot \delta^2 \cdot c_\gamma^2\}.
	\end{equation*}
\end{proposition}
To establish close coupling, we first notice that $q^w$ is reversible (Lemma \ref{lemma:positive}'s proof). Given this observation, applying \cite[Lemma 19]{alpw2022a} gives the following.
\begin{lemma} \label{lemma:pidealclosecoupling}
    Let $P_\infty$ be the locally balanced ideal kernel with proposal $q^w$. Let $\alphid(x):=\int q^w(x,\dif y)\alphid(x,y)$ denote the corresponding average acceptance probability, and let $\overline{\alphid}:=\inf \alphid(x)$. For all $x,y\in\Xspace$, 
    \begin{equation*}
	   \tvdist{\ideal(x,\cdot)-\ideal(y,\cdot)} \leq \tvdist{q^w(x,\cdot)-q^w(y,\cdot)} + 1 - \overline{\alphid}, ,
	\end{equation*}
\end{lemma}
Thus, to establish the close coupling condition we seek bounds on the quantity $\tvdist{q^w(x,\cdot)-q^w(y,\cdot)}$, and the acceptance probability $\alpha_\infty$. 
\begin{lemma} \label{lemma:qwclosecoupling}
    Assume Assumption \ref{hp:pigaussian}. Then, we have 
    \begin{equation*}
    	\tvdist{q^w(x,\cdot)-q^w(y,\cdot)} \leq |x-y|\cdot \sqrt{\frac{1}{2\cdot\sigma^2\cdot(2+\sigma^2)}}
	\end{equation*}
\end{lemma}
\begin{proof}
   From \eqref{eq:qwisgaussian}, using standard results on the relative entropy between Gaussians with common covariances and different means, we deduce that for all $x,y\in\Xspace$,
   \begin{equation*}
   	\textup{KL}(q^w(x,\cdot),q^w(y,\cdot)) = |x-y|^2 \cdot \left(\frac{2}{2+\sigma^2}\right)^2 \cdot \frac{2+\sigma^2}{4\sigma^2} =|x-y|^2\cdot\frac{1}{\sigma^2 \cdot (2+\sigma^2)}.
   \end{equation*}
   Pinsker's inequality concludes the proof. 
\end{proof}
Lemma \ref{lemma:pidealclosecoupling} and \ref{lemma:qwclosecoupling} immediately give the following.
\begin{corollary} \label{cor:pidealclosecoulping}
    $\ideal$ is a $(|\cdot|,\overline{\alphid}\cdot \sqrt{\sigma^2\cdot(2+\sigma^2)}\cdot 2^{-1/2},2^{-1}\cdot \overline{\alphid})$-close coupling.
\end{corollary}
We are now left with the task of deriving a lower bound to $\alphid$.
\begin{lemma}[Lower bound for $\alphid$] \label{lemma:alphidlb}
    Assume Assumption \ref{hp:pigaussian}. It holds $\overline{\alphid}\geq \frac{1}{2}\exp\left(-\frac{d\sigma^4}{4(2+\sigma^2)^2}\right)$. In particular, if $\sigma=\zeta\cdot d^{-1/4}$ for some $\zeta>0$, then $\overline{\alphid}\geq \frac{1}{2}\exp(-\zeta^4/16)$.
\end{lemma}
\begin{proof}
    See Appendix \ref{app:alphidlb}.
\end{proof}

\subsubsection{Spectral gap upper bound}
The main result of this section is as follows.
\begin{proposition} \label{prop:idealspgapub}
    Assume Assumption \ref{hp:pigaussian}. We have
    \begin{align*}
	   \gamma(\ideal) \leq \bigg( \frac{3}{2}\cdot\frac{\sigma^2}{2+\sigma^2}\bigg )\wedge \bigg(1+\frac{\sigma^4}{(2+\sigma^2)^2}\bigg)^{-d/2}
	\end{align*}
\end{proposition}
\begin{proof}
    This bound follows from Lemmas \ref{lemma:idealspgapubdiri} and  \ref{lemma:idealspgapubconductance}. 
\end{proof}	
To obtain this result, we combine two methods to derive bounds on the spectral gap. The first is based on its variational representation for reversible chains via its link with the Dirichlet form. 
\begin{lemma} \label{lemma:idealspgapubdiri}
	Assume Assumption \ref{hp:pigaussian}. It holds $\gamma(\ideal) \leq  \frac{3}{2}\cdot \frac{\sigma^2}{2+\sigma^2}$.
\end{lemma}
\begin{proof}
    Following Lemma 47 in \cite{alpw2022a}, we begin by noting
	\begin{equation} \label{eq:spgapvariational}
		\gamma(\ideal) = \inf_{f\in\mathrm{L}_0^2(\pi)} \frac{\cal{E}(\ideal,f)}{\norm{f }_{2,\pi}^2}.
	\end{equation}
	Consider $f(x):=\iprod{v}{x-\Ebb{X}}\in\mathrm{L}_0^2(\pi)$, where $v\in\r^{d}$ and the expectation is taken with respect to ~$\pi$;  we can check that $\norm{f }_{2,\pi}^2=|v|^2$. If we let $\rho:=2/(2+\sigma^2)$, then using \eqref{eq:qwisgaussian} we can compute 
    \begin{align*}
	   \cal{E}(\ideal,f) &= \frac{1}{2}\int \pi(\dif x) \ideal(x,\dif y)(f(x)-f(y))^2 \\
       &\leq \frac{1}{2}\int \pi(\dif x) q^w(x,\dif y)\iprod{v}{y-x}^2 \\
        &= \frac{1}{2} \int \pi(\dif x) \cal{N}(\dif y;x\rho,I_d \cdot\sigma^2\rho)\iprod{v}{y-x}^2 \\
        &\leq \frac{1}{2} \int \pi(\dif x) \cal{N}(\dif y;x\rho - x,I_d \cdot \sigma^2\rho)\iprod{v}{y}^2 \\
        &= \frac{1}{2} \int \pi(\dif x)[|v|^2\sigma^2\rho+\iprod{v}{x}^2(1-\rho)^2] = \frac{|v|^2}{2}[\sigma^2\rho + (1-\rho)^2];
	\end{align*}
    we conclude substituting for $\rho$ and applying the variational formula.
\end{proof}
The second method is based instead on the notion of conductance and the Cheeger inequalities from Lemma \ref{lemma:cheeger}.
\begin{lemma} \label{lemma:idealspgapubconductance}
    Assume Assumption \ref{hp:pigaussian}. It holds $\gamma(\ideal) \leq \left(1+\frac{\sigma^4}{(2+\sigma^2)^2}\right)^{-d/2}$.
\end{lemma}
\begin{proof}
    If we let $\rho:=2/(2+\sigma^2)$, using the expressions \eqref{eq:qwisgaussian} and \eqref{eq:alphainf},
    \begin{align*}
	   \alphid(x)&=\int \cal{N}(\dif y;x\rho;I_d \cdot\sigma^2\rho) \min\left(1,\exp(-\psi(y)+\psi(x))\right) \\
	   &\leq \exp(\psi(x)) \int \cal{N}(\dif y;x\rho;I_d \cdot\sigma^2\rho)	\exp(-\psi(y))  \\
	   &= \bigg(1-2t\bigg)^{-d/2} \cdot \exp\bigg(|x|^2\cdot\bigg\{\frac{\sigma^2}{4(2+\sigma^2)}+\frac{\rho^2t}{1-2t}\bigg\}\bigg).
	\end{align*}
	where $t:=-\sigma^4/4(2+\sigma^2)^2$, as a moment-generating function of a non-central chi-square. Define the sets
    \begin{equation*}
	    A_{\epsilon}:=\bigg\{x:\bigg\{\frac{\sigma^2}{4(2+\sigma^2)}+\frac{\rho^2t}{1-2t}\bigg\}\leq \epsilon^2 \bigg\}.
	\end{equation*}
    By Lemma \ref{lemma:cheeger}, we have
	\begin{align*}
	   \gamma(\ideal)&\leq \Phi(\ideal) \leq \frac{(\pi\otimes \ideal)(A_\epsilon\times A_\epsilon^\complement)}{\pi(A_\epsilon)} \\ &\leq \int \pi\rvert_{A_\epsilon}(\dif x)\alphid(x)\leq \bigg(1+\frac{\sigma^4}{(2+\sigma^2)^2}\bigg)^{-d/2}\exp(\epsilon^2)
	\end{align*}
    and the claim then follows upon taking the $\liminf$ as $\epsilon$ approaches $0$.
\end{proof}
\section{Conclusions}
In this paper, we analysed the Multiple-try Metropolis algorithm using Poincaré inequalities, interpreting it as an auxiliary-variable approach to a resampling-based approximation of an idealized Metropolis-–Hastings algorithm.

We investigated a limitation of Multiple-try approaches: given a Multiple-try Metropolis algorithm with $n-1$ proposals, introducing  one additional proposal can enhance performance by at most a factor of $n/(n-1)$ with respect to the spectral gap. This reinforces the argument that Multiple-try is appropriate (only) in contexts involving parallel computation. For the commonly used `globally-balanced' implementation, we showed that when paired with random walk proposals, the spectral gap diminishes as $n$ grows, underscoring its unsuitability even in parallel computing settings. These findings suggest that globally balanced Multiple-try algorithms are best avoided in general.

We derived comparison results between Multiple-try kernels and the idealized Metropolis–Hastings algorithm, framed in terms of the moments of the importance weights. Tighter bounds arise when these importance weights admit finite moments of high order. For Gaussian target distributions, we derived explicit spectral gap estimates for the idealized Metropolis–Hastings algorithm using locally-balanced weights \cite{gagnon2023,chang2022}. Notably, the spectral gap of the algorithm scales as $d^{-1/2}$, which is better than the $d^{-1}$ scaling of plain Random Walk Metropolis, and suggests that in the large $n$ limit, the Multiple-try strategy employing locally-balanced weights might achieve good performance in higher dimensional problems compared to more standard algorithms. Combining the comparison results and the spectral estimates, we obtained non-asymptotic convergence bounds for Multiple-try Metropolis under Gaussian targets, expanding on prior work limited to independence samplers \cite{yang2023}. As reflected in the paucity of available results in the literature, the Multiple-try algorithm is particularly complicated to analyse due its complex structural properties, and the comparison approach described in this paper is one approach to tackle such difficulties. On the other hand, the Gaussian assumptions is clearly restrictive. We expect a similar results to hold, with suitable adaptations, for more general strongly log-concave and gradient Lipschitz target measures. However, in this work, we make use in several places of the Gaussian structure of $\pi$, so the generalisation does not appear immediate regardless. Even in the Gaussian case, the rates of convergence which we obtain are slower-than-exponential; given what is available for the basic Random Walk Metropolis in this case, one expects that there should be room for further refinements of these results. 

\printbibliography

@article{gagnon2023,
  title={Improving multiple-try Metropolis with local balancing},
  author={Gagnon, Philippe and Maire, Florian and Zanella, Giacomo},
  journal={Journal of Machine Learning Research},
  volume={24},
  number={248},
  pages={1--59},
  year={2023}
}

@article{martino2017,
  title={Issues in the multiple try Metropolis mixing},
  author={Martino, Luca and Louzada, Francisco},
  journal={Computational Statistics},
  volume={32},
  pages={239--252},
  year={2017},
}

@article{alpw2022a,
  title={Explicit convergence bounds for Metropolis Markov chains: isoperimetry, spectral gaps and profiles},
  author={Andrieu, Christophe and Lee, Anthony and Power, Sam and Wang, Andi Q},
  journal={The Annals of Applied Probability},
  year={2024},
  volume={34},
  number={4},
  pages={4022--4071}
}

@article{yang2023,
  title={Convergence rate of multiple-try Metropolis independent sampler},
  author={Yang, Xiaodong and Liu, Jun S},
  journal={Statistics and Computing},
  volume={33},
  number={4},
  pages={79},
  year={2023},
  publisher={Springer}
}

@article{alpw2022b,
  title={Comparison of Markov chains via weak Poincar{\'e} inequalities with application to pseudo-marginal MCMC},
  author={Andrieu, Christophe and Lee, Anthony and Power, Sam and Wang, Andi Q},
  journal={The Annals of Statistics},
  volume={50},
  number={6},
  pages={3592--3618},
  year={2022},
}

@article{chang2022,
  title={Rapidly mixing multiple-try Metropolis algorithms for model selection problems},
  author={Chang, Hyunwoong and Lee, Changwoo and Luo, Zhao Tang and Sang, Huiyan and Zhou, Quan},
  journal={Advances in Neural Information Processing Systems},
  volume={35},
  pages={25842--25855},
  year={2022}
}

@article{liu2000,
	title={The multiple-try method and local optimization in Metropolis sampling},
	author={Liu, Jun S and Liang, Faming and Wong, Wing Hung},
	journal={Journal of the American Statistical Association},
	volume={95},
	pages={121--134},
	year={2000}
}

@article{alpw2022c,
	title={Poincar{\'e} inequalities for Markov chains: a meeting with Cheeger, Lyapunov and Metropolis},
	author={Andrieu, Christophe and Lee, Anthony and Power, Sam and Wang, Andi Q},
	journal =	 {arXiv:2208.05239},
	year={2022}
}

@article{zanella2020,
  title={Informed proposals for local MCMC in discrete spaces},
  author={Zanella, Giacomo},
  journal={Journal of the American Statistical Association},
  year={2020},
  pages={852--865},
  volume={115}
}

@article{metropolis1953,
	title={Equation of state calculations by fast computing machines},
	author={Metropolis, Nicholas and Rosenbluth, Arianna W and Rosenbluth, Marshall N and Teller, Augusta H and Teller, Edward},
	journal={Journal of Chemical Physics},
	volume={21},
	pages={1087--1092},
	year={1953},
}

@article{hastings1970,
	title={Monte {C}arlo sampling methods using {M}arkov chains and their applications},
	author={Hastings, Keith W},
	journal={Biometrika},
	volume={57},
	pages={97--109},
	year={1970},
}

@book{frenkel1996,
	title={Understanding molecular simulation: from algorithms to applications},
	author={Frenkel, Daan and Smit, Berend and Ratner, Mark A},
	year={1996},
	publisher={Academic Press}
}

@article{gelman1997,
	title={Weak convergence and optimal scaling of random walk Metropolis algorithms},
	author={Gelman, Andrew and Gilks, Walter R and Roberts, Gareth O},
	journal={The Annals of Applied Probability},
	volume={7},
	number={1},
	pages={110--120},
	year={1997},
	publisher={Institute of Mathematical Statistics}
}

@article{bedard2012,
	title={Scaling analysis of multiple-try MCMC methods},
	author={B{\'e}dard, Myl{\`e}ne and Douc, Randal and Moulines, Eric},
	journal={Stochastic Processes and their Applications},
	volume={122},
	number={3},
	pages={758--786},
	year={2012},
	publisher={Elsevier}
}

@article{casarin2013,
	title={Interacting multiple try algorithms with different proposal distributions},
	author={Casarin, Roberto and Craiu, Radu and Leisen, Fabrizio},
	journal={Statistics and Computing},
	volume={23},
	pages={185--200},
	year={2013},
	publisher={Springer}
}

@article{pozza2024,
	title={On the fundamental limitations of multiproposal Markov chain Monte Carlo algorithms},
	author={Pozza, Francesco and Zanella, Giacomo},
	journal={arXiv:2410.23174},
	year={2024}
}

@article{bornn2017,
	title={The use of a single pseudo-sample in approximate Bayesian computation},
	author={Bornn, Luke and Pillai, Natesh S and Smith, Aaron and Woodard, Dawn},
	journal={Statistics and Computing},
	volume={27},
	pages={583--590},
	year={2017},
	publisher={Springer}
}

@article{martino2018,
	title={A review of multiple try MCMC algorithms for signal processing},
	author={Martino, Luca},
	journal={Digital Signal Processing},
	volume={75},
	pages={134--152},
	year={2018},
	publisher={Elsevier}
}

@article{craiu2007,
	title={Acceleration of the multiple-try Metropolis algorithm using antithetic and stratified sampling},
	author={Craiu, Radu V and Lemieux, Christiane},
	journal={Statistics and computing},
	volume={17},
	pages={109--120},
	year={2007},
	publisher={Springer}
}

@inproceedings{pandolfi2010,
	title={A generalization of the Multiple-try Metropolis algorithm for Bayesian estimation and model selection},
	author={Pandolfi, Silvia and Bartolucci, Francesco and Friel, Nial},
	booktitle={Proceedings of the 13th International Conference on Artificial Intelligence and Statistics},
	pages={581--588},
	year={2010},
}

@article{Doucet2015,
  title={Efficient implementation of Markov chain Monte Carlo when using an unbiased likelihood estimator},
  author={Doucet, Arnaud and Pitt, Michael K and Deligiannidis, George and Kohn, Robert},
  journal={Biometrika},
  volume={102},
  number={2},
  pages={295--313},
  year={2015},
  publisher={Oxford University Press}
}

@article{Lovasz1999,
  title={Hit-and-run mixes fast},
  author={Lov{\'a}sz, L{\'a}szl{\'o}},
  journal={Mathematical programming},
  volume={86},
  pages={443--461},
  year={1999},
  publisher={Springer}
}

@article{Vempala2005,
  title={Geometric random walks: a survey},
  author={Vempala, Santosh},
  journal={Combinatorial and Computational Geometry},
  volume={52},
  number={573-612},
  pages={2},
  year={2005},
  publisher={Cambridge University Press Cambridge}
}

@article{Chen2020,
  title={Fast mixing of Metropolized Hamiltonian Monte Carlo: Benefits of multi-step gradients},
  author={Chen, Yuansi and Dwivedi, Raaz and Wainwright, Martin J and Yu, Bin},
  journal={Journal of Machine Learning Research},
  volume={21},
  number={92},
  pages={1--72},
  year={2020}
}
\begin{appendix}
\section{Proofs for Section \ref{sec:comparisons}} 

\subsection{Useful expressions } \label{sec:usefulexpressions}
The proofs involving the Multiple-try Metropolis kernels (Algorithm~\ref{alg:mtm}) are necessarily heavy in notation due the different mechanisms involved. In this section we collect some useful expressions to follow the upcoming proofs.
In all the expressions below, we let $(Y^{[n]},Z^{[n]})|(X=x,Y_1=y)\sim\delta_{y}(\dif y_1)\delta_{x}(\dif z_1)\prod_{i=2}^{n} q(x,\dif y_i)q(y,\dif z_i)$ and we denote
\begin{align*}
	\varpi(x,y)&:=\frac{\dif q^w(x,\cdot)}{\dif q(x,\cdot)}(y)=\frac{w(x,y)}{(qw)(x)},\\
	(\widetilde{qw})_n(x,y)&:=\Ebb{(\widehat{qw}_n)(x,Y^{[n]})^{-1}|Y_1=y}^{-1},\\
	(\widehat{qw}_n)(x,Y^{[n]})&:=n^{-1}\sum_{i=1}^n w(x,Y_i).
\end{align*}
Because both the semi-ideal kernel $\ir{n}$ and the ideal kernel $\ideal$ are simply Metropolis--Hastings chains with different proposals, $q^w$ \eqref{eq:qw} and  $ \widetilde{q}^w_n$ \eqref{eq:qwn}, respectively, their transition probabilities are readily written as
\begin{align*}
	&\ir{n}(x,A\backslash \{x\}) = \int_{A\backslash \{x\}} \frac{q(x,\dif y)\cdot w(x,y)}{(\widetilde{qw}_n)(x,y)} \cdot \alphir{n}(x,y), \quad \text{where} \\
    &\alphir{n}(x,y):=\min\bigg\{1,\frac{\pi(y)\cdot q(y,x)\cdot w(y,x)\cdot (\widetilde{qw}_n)(x,y)}{\pi(x)\cdot q(x,y)\cdot w(x,y)\cdot (\widetilde{qw}_n)(y,x)}\bigg\}. &&
\end{align*}
\begin{align*}
    &\ideal(x,A\backslash \{x\}) = \int_{A\backslash \{x\}} \frac{q(x,\dif y)\cdot w(x,y)}{(qw)(x)} \cdot \alphid(x,y), \quad \text{where} \\
    &\alphid(x,y):=\min\bigg\{1,\frac{\pi(y)\cdot q(y,x)\cdot w(y,x)\cdot (qw)(x)}{\pi(x)\cdot q(x,y)\cdot w(x,y)\cdot (qw)(y)}\bigg\}. &&
\end{align*}
The Multiple-try chain is, however, not a straightforward Metropolis--Hastings algorithm. To derive its transition probabilities $\mtm{n}$, recall that the Multiple-try accepts a proposal $Y_I=Y$ with probability
\begin{equation*}
	\alphmtm{n}(x,Y^{[n]},Z^{[n]}):=\min\bigg(1,\frac{\pi(Y)\cdot q(Y,x)\cdot w(Y,x) \cdot (\widehat{qw}_n)(x,Y^{[n]})}{\pi(x)\cdot q(x,Y)\cdot w(x,Y) \cdot (\widehat{qw}_n)(Y,Z^{[n]})}\bigg),
\end{equation*}
Then, since each proposal is drawn i.i.d., for an independently drawn uniform random variable $U$, one has
\begin{align*}
	&\mtm{n}(x,A\backslash \{x\}) \\
    &=  \sum_{i=1}^n\mathbb{E}_x\left [\mathbf{1}{\left\{Y_i\in A\backslash \{x\}\right \}}\cdot \mathbf{1}{\{I=i\}} \cdot\mathbf{1}{\{Z_i=x\}} \cdot \mathbf{1}{\{U\le  \alphmtm{n}(x,Y^{[n]},Z^{[n]})\}}\right ] \\
    &= n \cdot \mathbb{E}_x\left [\mathbf{1}{\left \{Y_1\in A\backslash \{x\}\right \}}\cdot \mathbf{1}{\{I=1\}}\cdot \mathbf{1}{\{Z_1=x\}}\cdot \alphmtm{n}(x,Y^{[n]},Z^{[n]}) \right ]  \\
	&= n  \int \prod_{i=1}^n q(x,\dif y_i) \cdot \frac{w(x,y_1)\mathbf{1}{\left \{y_1\in A\backslash \{x\}\right \}}}{(\widehat{qw}_n)(x,y^{[n]})} \int \delta_{x}(\dif z_1)  \prod_{i=2}^n q(y_1,\dif z_i)\cdot\alphmtm{n}(x,y^{[n]},z^{[n]})
\end{align*}
Therefore, 
\begin{align*} 
	&\mtm{n}(x,A\backslash \{x\}) \\
    &= \int_{A\backslash \{x\}} q(x,\dif y)\cdot w(x,y) \cdot 
	 \mathbb{E}\left[\frac{1}{(\widehat{qw}_n)(x,Y^{[n]})}\cdot\alphmtm{n}(x,Y^{[n]},Z^{[n]}) \middle | Z_1=x,Y_1=y\right]. 
\end{align*}
\subsection{Proof of Proposition \ref{prop:spidifferentn}} \label{subsec:pf_MTM_neg}
By definition of $(\widetilde{qw}_n)(x,y)$, we can write $\mtm{n}(x,A\backslash \{x\})$ as 
\begin{align*}
    &\int_{A\backslash \{x\}} q(x,\dif y)\cdot w(x,y) \cdot  \\
    &\mathbb{E}\left[\min\bigg\{ \frac{n}{\sum_{i=1}^n w(x,Y_i)},\frac{\pi(y)\cdot q(y,x) \cdot w(y,x) \cdot n}{\pi(x)\cdot q(x,y)\cdot w(x,y)\cdot \sum_{i=1}^n w(y,Z_i)}\bigg\} \middle | Z_1=x,Y_1=y\right].
\end{align*}
However, since the weights are non-negative, $\sum_{i=1}^n w_i \geq \sum_{i=1}^{n-1} w_i$, hence
\begin{align*}
	&\min\bigg\{ \frac{n}{\sum_{i=1}^n w(x,Y_i)},\frac{\pi(y)\cdot q(y,x) \cdot w(y,x) \cdot n}{\pi(x)\cdot q(x,y)\cdot w(x,y)\cdot \sum_{i=1}^n w(y,Z_i)}\bigg\} \\
	&\leq \frac{n}{n-1} \min\bigg\{ \frac{n-1}{\sum_{i=1}^{n-1} w(x,Y_i)},\frac{\pi(y)\cdot q(y,x) \cdot w(y,x) \cdot (n-1)}{\pi(x)\cdot q(x,y)\cdot w(x,y)\cdot \sum_{i=1}^{n-1} w(y,Z_i)}\bigg\}
\end{align*}
and the result follows upon taking expectations with respect to~$(Y^{[n]},Z^{[n]})|(x,y)$, integrating with respect to~$q(x,\dif y)\cdot w(x,y)$ and finally applying Lemma \ref{lemma:wpifromdomination}. A similar proof works for the $\widetilde{P}_n$ kernels.
\qed

\subsection{Proofs of the comparison inequalities} 

\subsubsection{Proof of Lemma \ref{lemma:idealirspi}} \label{app:idealirspi}

Using the inequality $\min\{1,a\cdot b\}\geq \min\{1,a\}\cdot\min\{1,b\}$ for $a,b\geq 0$ we write
\begin{align*}
	\alphir{n}(x,y)&=\min\bigg\{ 1,\frac{\pi(y)\cdot q(y,x) \cdot w(y,x) \cdot (qw)(x)}{\pi(x)\cdot q(x,y) \cdot w(x,y) \cdot (qw)(y)}\frac{(\widetilde{qw}_n)(x,y)\cdot (qw)(y)}{(\widetilde{qw}_n)(y,x)\cdot (qw)(x)}\bigg\} \\
	&\geq \alphid(x,y) \min\bigg\{1,\frac{(\widetilde{qw}_n)(x,y)\cdot (qw)(y)}{(\widetilde{qw}_n)(y,x)\cdot (qw)(x)}\bigg\}
\end{align*}
from which it follows that, off-diagonal,
\begin{align} \label{eq:idealirspitempline}
	\ir{n}(x,\dif y) &\geq \frac{q(x,\dif y)\cdot w(x,y)}{(qw)(x)} \cdot \alphid(x,y) \cdot \min\bigg\{ \frac{(qw)(x)}{(\widetilde{qw}_n)(x,y)},\frac{(qw)(y)}{(\widetilde{qw}_n)(y,x)}\bigg\} \nonumber \\
	&= \ideal(x,\dif y) \cdot \min\bigg\{ \frac{(qw)(x)}{(\widetilde{qw}_n)(x,y)},\frac{(qw)(y)}{(\widetilde{qw}_n)(y,x)}\bigg\} 
\end{align}
The claim then follows by the definition of $(\widetilde{qw}_n)(x,y)$ and the uniform bound.
\subsubsection{Proof of Proposition \ref{pr:idealirwpi}} \label{app:idealirwpi}
Equation \eqref{eq:idealirspitempline} shows that for any $(x,A)\in\Xspace\times\mathcal{F}$, it holds that $\ir{n}(x,A \backslash \{x\}) \geq \int_{A \backslash \{x\}} \eta_n(x,y)\ideal(x,\dif y)$, with
\begin{equation*}
	\eta_n(x,y):=\min\bigg\{ \frac{(qw)(x)}{(\widetilde{qw}_n)(x,y)},\frac{(qw)(y)}{(\widetilde{qw}_n)(y,x)}\bigg\}.
\end{equation*} 
Lemma \ref{lemma:wpifromdomination} then shows that
\begin{equation} \label{eq:WPIIRID1}
	\cal{E}(\ideal,f)\leq s\cal{E}(\ir{n},f) + 2^{-1}\oscnorm{f}^2 (\pi\otimes \ideal)(A_n(s)^\complement \cap \{X\neq Y\}) 
\end{equation}
where $A_n(s):=\{(x,y)\in\Xspace^2: \eta_n(x,y) > 1/s \}$. Now, by the union bound, symmetry, and finally by the fact that $\alpha_\infty\leq 1$, we have
\begin{align*}
	&(\pi\otimes \ideal)(A_n(s)^\complement \cap \{X\neq Y\}) \\
	&\leq 	(\pi\otimes \ideal)\bigg( \bigg\{\frac{(qw)(X)}{(\widetilde{qw}_n)(X,Y)} <\frac{1}{s} \bigg\} \cap \{X\neq Y\}\bigg)  \\ &+ (\pi\otimes \ideal)\bigg(\bigg\{\frac{(qw)(Y)}{(\widetilde{qw}_n)(Y,X)} <\frac{1}{s} \bigg\}\cap \{X\neq Y\}\bigg) \\
	&\leq 2\cdot (\pi\otimes \ideal)\bigg( \bigg\{ \frac{(qw)(X)}{(\widetilde{qw}_n)(X,Y)} <\frac{1}{s}\bigg\}\cap \{X\neq Y\}\bigg) \\
	&\leq 2 \cdot (\pi\otimes q^w) \bigg( \bigg\{ \frac{(qw)(X)}{(\widetilde{qw}_n)(X,Y)} <\frac{1}{s}\bigg\} \bigg) \\
    &= 2 \cdot (\pi\otimes q^w) \bigg( \mathbb{E}\bigg[\frac{1}{n^{-1}\sum_{i=1}^n \varpi(X,Y_i)} \mid  Y_1=Y\bigg] < \frac{1}{s}\bigg) 
\end{align*}
where the expectation is taken with respect to $Y^{[n]}|Y_1=y\sim\delta_{y}(\dif y_1)\prod_{i=2}^{n} q(x,\dif y_i)$.
Moreover, for any $p\in[1,\infty)$, by Markov's and Jensen's inequalities (twice), we write
\begin{align*}
	&(\pi\otimes \ideal)(A_n(s)^\complement \cap \{X\neq Y\}) \\
    &\leq s^{-p}\cdot (\pi\otimes q^w) \bigg( \mathbb{E}\bigg[\frac{1}{n^{-1}\sum_{i=1}^n \varpi(X,Y_i)} \mid  Y_1=Y\bigg]^{-p}\bigg) \\
    &\leq s^{-p}\cdot (\pi\otimes q^w) \bigg(\mathbb{E}\bigg[\bigg\{\frac{1}{n}\sum_{i=1}^n\varpi(X,Y_i)\bigg\}^p\mid Y_1=Y \bigg]\bigg) \\
    &\leq s^{-p}\cdot (\pi\otimes q^w) \bigg( \frac{1}{n}\sum_{i=1}^n \mathbb{E}[\varpi(X,Y_i)^p\mid Y_1=Y] \bigg)\\
    &= s^{-p}\cdot\bigg(\frac{1}{n}(\pi\otimes q^w)(\varpi(X,Y)^p) + \frac{1}{n}\sum_{i=2}^n (\pi\otimes q) (\varpi(X,Y_i)^p) \bigg),
\end{align*}
and the final expression follows by the definition of $\varpi$ and $q^w$.
\qed

\subsubsection{Proof of Lemma \ref{lemma:irmtmspi}} \label{app:irmtmspi}

Let $(Y^{[n]},Z^{[n]})|(X=x,Y_1=y)\sim\delta_{y}(\dif y_1)\delta_{x}(\dif z_1)\prod_{i=2}^{n} q(x,\dif y_i)q(y,\dif z_i)$. Let
\begin{equation*}
	\zeta_n(x,y):=\Ebb{ \min\left(\frac{(\widetilde{qw})_n(x,y)}{(\widehat{qw}_n)(x,Y^{[n]})}, \frac{(\widetilde{qw})_n(y,x)}{(\widehat{qw}_n)(y,Z^{[n]})} \right) \middle |  Y_1=y,Z_1=x }
\end{equation*} 
Having in mind the expressions in Section \ref{sec:usefulexpressions}, and similarly to the start of the proof in Section \ref{app:idealirspi}, we can write for any $x\in\Xspace$
\begin{align} \label{eq:irmtmspitempline}
	& \mtm{n}(x,A\backslash \{x\}))  \nonumber  \\
	&\geq \int_{A\backslash \{x\}} \frac{q(x,y)\cdot w(x,y)}{(\widetilde{qw})_n(x,y)}\cdot\alphir{n}(x,y) \cdot \nonumber \\
	&\cdot \Ebb{\frac{(\widetilde{qw})_n(x,y)}{(\widehat{qw}_n)(x,Y^{[n]})}\cdot \min\left(1,\frac{(\widehat{qw}_n)(x,Y^{[n]})}{(\widetilde{qw})_n(x,y)} \cdot\frac{(\widetilde{qw})_n(y,x)}{(\widehat{qw}_n)(y,Z^{[n]})} \right) \middle |  Y_1=y,Z_1=x } \nonumber \\	
	&= \int_{A \backslash \{x\}} \zeta_n(x,y)\cdot \ir{n}(x,\dif y) \\
	&\geq |\varpi|_\infty^{-2}\cdot|\varpi^{-1}|_\infty^{-2}\cdot \ir{n}(x,A\backslash \{x\})  .\nonumber 
\end{align}
\qed
\subsubsection{Proof of Proposition \ref{lemma:irmtmwpi}} \label{app:irmtmwpi}
Equation \eqref{eq:irmtmspitempline} shows that for any $(x,A)\in\Xspace\times\mathcal{F}$, it holds that $\mtm{n}(x,A \backslash \{x\}) \geq \int_{A \backslash \{x\}} \zeta_n(x,y)\ir{n}(x,\dif y)$.
Let $(X,Y)\sim\pi\otimes\ir{n}$. On the one hand, Lemma \ref{lemma:wpifromdomination} shows that
\begin{equation} \label{eq:WPIIRID1}
	\cal{E}(\ir{n},f)\leq s\cdot \cal{E}(\mtm{n},f) + 2^{-1}\oscnorm{f}^2 (\pi\otimes \ir{n})(B_n(s)^\complement \cap \{X\neq Y\}) 
\end{equation}
where $B_n(s):=\{(x,y)\in\Xspace^2: \zeta_n(x,y) > 1/s \}$. On the other hand, by Markov's inequality,
\begin{align*}
	(\pi\otimes \ir{n})(B_n(s)^\complement \cap \{X\neq Y\}) \leq s^{-p}\cdot (\pi\otimes \ir{n})(\zeta_n(X,Y)^{-p}\cap \{X\neq Y\}).
\end{align*}
At this point, using the elementary relation $\min\{a,b\}^{-1} = \max \left\{ a^{-1},b^{-1} \right\}$, we obtain that
\begin{align*} 
	&(\pi\otimes \ir{n})(B_n(s)^\complement \cap \{X\neq Y\}) \nonumber \\
	&\leq s^{-p}\cdot (\pi\otimes\ir{n})\bigg(\mathbb{E}\left[\max\bigg\{\frac{(\widehat{qw}_n)(Y,Z^{[n]})}{(\widetilde{qw})_n(y,x)},\frac{(\widehat{qw}_n)(X,Y^{[n]})}{(\widetilde{qw})_n(x,y)}	
	\bigg\}\middle | Y_1=Y, Z_1=X\right]^p \bigg)  \nonumber \\
	&\leq  s^{-p} \cdot (\pi\otimes\ir{n})\bigg(\mathbb{E}\left [\frac{(\widehat{qw}_n)(Y,Z^{[n]})}{(\widetilde{qw})_n(Y,X)} + \frac{(\widehat{qw}_n)(X,Y^{[n]})}{(\widetilde{qw})_n(X,Y)}	 \middle |  Y_1=Y, Z_1=X\right]^p\bigg) \nonumber \\
	&\leq s^{-p} \cdot 2^{p-1}\cdot (\pi\otimes\ir{n})\bigg(\mathbb{E}\left[\frac{(\widehat{qw}_n)(X,Y^{[n]})}{(\widetilde{qw})_n(X,Y)}\middle | Y_1=Y\right ]^p + \mathbb{E}\left [\frac{(\widehat{qw}_n)(Y,Z^{[n]})}{(\widetilde{qw})_n(Y,X)}\middle | Z_1=X\right ]^p \bigg) 
\end{align*}
where in the last line we used the fact $(a+b)^p\leq 2^{p-1}(a^p+b^p)$. Because, under $\pi\otimes\ir{n}$, $Y$ is marginally distributed as $\pi$, the two summands above have identical expectations under $\pi\otimes\ir{n}$, and we thus have
\begin{align} \label{eq:inteqns0s2}
(\pi\otimes \ir{n})(B_n(s)^\complement \cap \{X\neq Y\}) 
 &\leq s^{-p} \cdot 2^{p}\cdot (\pi\otimes\ir{n})\bigg(\mathbb{E}\bigg[\frac{(\widehat{qw}_n)(X,Y^{[n]})}{(\widetilde{qw})_n(X,Y)}\mid Y_1=Y\bigg]^p \bigg) \nonumber \\
 &\leq s^{-p} \cdot 2^{p}\cdot (\pi\otimes\ir{n})\bigg(\mathbb{E}\bigg[ \bigg(\frac{(\widehat{qw}_n)(X,Y^{[n]})}{(\widetilde{qw})_n(X,Y)} \bigg)^{p}\mid Y_1=Y\bigg] \bigg).
\end{align}
Now we aim to control the expectation inside the right hand side above. We start with the denominator. Since, by Jensen's inequality, we have $\frac{n}{\sum_{i=1}^n w(x,Y_i)}\leq \frac{1}{n}\sum_{i=1}^n w(x,Y_i)^{-1}$, we can estimate
\begin{align*}
	(\widetilde{qw}_n)(x,y)^{-p} &= \mathbb{E}\left [\frac{n}{\sum_{i=1}^n w(x,Y_i)}\middle | Y_1=y \right ]^p 
    \leq \mathbb{E}\left [\frac{1}{n}\sum_{i=1}^n w(x,Y_i)^{-1}\middle | Y_1=y\right ]^p \nonumber \\
    &\leq \mathbb{E}\left [\frac{1}{n}\sum_{i=1}^n w(x,Y_i)^{-p}\middle | Y_1=y\right ] \nonumber \\
    &= \frac{1}{n}\cdot w(x,y)^{-p}+\frac{n-1}{n}\cdot \int w(x,y')^{-p}q(x,\dif y'),
\end{align*}
For the numerator, instead, just note that one application of Jensen's inequality yields
\begin{align*}
    \mathbb{E}[(\widehat{qw}_n)(x,Y^{[n]})^p] \leq \frac{1}{n}\cdot w(x,y)^p + \frac{n-1}{n}\int w(x,y')^pq(x,\dif y').
\end{align*}
Recalling that $\varpi(x,y):=w(x,y)/(qw)(x)$, using the above bounds and multiplying and dividing by $(qw)(x)$ yields
\begin{align*} 
	&\mathbb{E}\left [ \bigg(\frac{(\widehat{qw}_n)(x,Y^{[n]})}{(\widetilde{qw})_n(x,y)} \bigg)^{p}\middle | Y_1=y\right ] \\
	&\leq \left(  \frac{1}{n}\cdot \varpi(x,y)^p + \frac{n-1}{n}\int \varpi(x,y')^pq(x,\dif y') \right)\cdot \nonumber \\
    &\cdot\left(\frac{1}{n}\cdot \varpi(x,y)^{-p}+\frac{n-1}{n}\cdot \int \varpi(x,y')^{-p}q(x,\dif y')\right) \nonumber \\
    & \leq \frac{1}{2}\left(  \frac{1}{n}\cdot \varpi(x,y)^p + \frac{n-1}{n}\int \varpi(x,y)^pq(x,\dif y') \right)^2 \nonumber \\
    &+ \frac{1}{2}\left(\frac{1}{n}\cdot \varpi(x,y)^{-p}+\frac{n-1}{n}\cdot \int \varpi(x,y')^{-p}q(x,\dif y')\right)^2 \nonumber \\
	&\leq  \frac{1}{n^2}\cdot \left(  \varpi(x,y)^{2p} +  \varpi(x,y)^{-2p} \right) \nonumber \\
    &+ \frac{(n-1)^2}{n^2} \left( \int \varpi(x,y')^{2p} q(x,\dif y') + \int \varpi(x,y')^{-2p}q(x,\dif y')   \right)
\end{align*}
where the penultimate inequality follows by Young's inequality $ab\leq a^2/2+b^2/2$ and the final inequality by $(a+b)^2\leq 2a^2+2b^2$ and Jensen's inequality.
Upon integration in $\pi\otimes\ir{n}$ we thus get 
\begin{align} \label{eq:sa92msd}
	&(\pi\otimes\ir{n})\left(\Ebb{\frac{(\widehat{q w}_n)(X,Y^{[n]})}{(\widetilde{qw}_n)(Y,X)}\middle | Z_1=X}^p  \right) \nonumber \\
	&\leq  \frac{2}{n^2}\cdot \left( \pi\otimes\ir{n}(\varpi(X,Y)^{2p})+ \pi\otimes\ir{n}(\varpi(X,Y)^{-2p}) \right) \nonumber \\
    &+ \frac{2(n-1)^2}{n^2} \left( (\pi\otimes q) (\varpi(X,Y)^{2p}) + (\pi\otimes q) (\varpi(X,Y)^{-2p}) \right).
\end{align}
The first two summands  multiplied by $n^{-2}$ are negligible. In fact,
with the very crude bound $(\widetilde{q \varpi}_n)(x,y)^{-1}\leq n \cdot \varpi(x,y)^{-1}$, and by observing that $\varpi(x,x)=(qw)(x)^{-1}$,
\begin{align*}
	(\pi\otimes\ir{n})(\varpi(X,Y)^{2p}) 
	&\leq \int  \frac{\varpi(x,y)^{2p}w(x,y)}{(\widetilde{q w}_n)(x,y)}\pi(\dif x)q(x,\dif y) + \int (qw)(x)^{-1}\pi(\dif x) \\
    &= \int \frac{\varpi(x,y)^{2p+1}}{(\widetilde{q \varpi}_n)(x,y)}\pi(\dif x)q(x,\dif y) + \int (qw)(x)^{-1}\pi(\dif x) \\
	&\leq n \cdot (\pi\otimes q)(\varpi(X,Y)^{2p})  + \int (qw)(x)^{-1}\pi(\dif x)  \\
	&\leq n \cdot (\pi\otimes q)(\varpi(X,Y)^{2p})  + (\pi\otimes q)(w(X,Y)^{-1}) = K_2 \cdot n^{-1}
\end{align*}
where the last bound on the rightmost term follows from Jensen's inequality $(qw)(x)^{-1}\leq \int w(x,y)^{-1}\cdot q(\dif y)$. One can do analogous computation for $(\pi\otimes\ir{n})(\varpi(X,Y)^{-2p})$. Combining these estimates with \eqref{eq:sa92msd} shows
\begin{align*}
 &(\pi\otimes\ir{n})\bigg(\mathbb{E}\left [\frac{(\widehat{qw}_n)(X,Y^{[n]})}{(\widetilde{qw})_n(X,Y)}\middle | Y_1=Y\right ]^p \bigg) \\
 &\leq K_2\cdot n^{-1}+  \frac{2(n-1)^2}{n^2} \left( (\pi\otimes q) (\varpi(X,Y)^{2p}) + (\pi\otimes q) (\varpi(X,Y)^{-2p}) \right) \\
 &\leq  K_2\cdot n^{-1} +K_3\cdot n^{-2}+  2\cdot(M_\varpi(2p) +M_\varpi(2p))
\end{align*}
for a rolling constants $K_2,K_3$, and the final bound follows upon combining this with \eqref{eq:inteqns0s2}.

\qed

\section{Proofs for Section \ref{sec:idealanalysis}}
	
\subsection{Proof of Proposition \ref{prop:spectralgapgbmtm}} \label{app:spectralgapgbmtm}
	
Fix an arbitrary $\varepsilon>0$. By the definition of conductance, for any arbitrary set $A\in\cal{F}$, we can write
\begin{equation*}
    \Phi(\mtm{n}) \leq \frac{(\pi\otimes \mtm{n})(A\times A^\complement)}{\pi(A)} = \int \mtm{n}(x,A^\complement)\pi\rvert_{A}(\dif x)  \leq \int \alphmtm{n}(x) \pi\rvert_{ A}(\dif x),
\end{equation*}
where $\pi\rvert_{A}$ denotes the restriction of $\pi$ to $A$. Write $\alphmtm{n}(x)\leq \alphid(x) + |\alphid(x)-\alphmtm{n}(x)|$. 
By \cite[Proposition 3]{gagnon2023}, $\alphid(x)\leq \exp\bigg(-|x|^2\cdot c_{1,\sigma} \bigg)\cdot c_{2,\sigma}$ with
\begin{equation*}
     c_{1,\sigma}:=\frac{\sigma^2 }{2((1+\sigma^2)^2-\sigma^2)}, \quad c_{2,\sigma}=\bigg(1-\frac{\sigma^2}{(1+\sigma^2)^2}\bigg)^{-d/2}
\end{equation*}
Let $R_\varepsilon>0$ be such that $\exp(-|x|^2\cdot c_{1,\sigma})\cdot c_{2,\sigma}<\varepsilon$ for all $|x|>R_\varepsilon$. If we then define the set $A_{R_\varepsilon}:=\{x\in\mathsf{E}:|x|>R_\varepsilon \}$ from the previous estimates we have
\begin{equation*}
    \Phi(\mtm{n}) \leq \varepsilon  + \int  |\alphid(x)-\alphmtm{n}(x)|\pi\rvert_{A_{R_\varepsilon}}(\dif x).
\end{equation*}
and because, as we show below, $|\alphid(x)-\alphmtm{n}(x)|\rightarrow 0$ for all $x\in\mathsf{E}$, the bounded convergence theorem shows that the rightmost integral can be made lower than $\varepsilon$ for $n$ sufficiently large. It follows that $\Phi(\mtm{n}) \leq 2\cdot \varepsilon$ for $n$ large. We are only left to show  $|\alphid(x)-\alphmtm{n}(x)|\rightarrow 0$ for all $x\in\mathsf{E}$. We note that $(\widehat{qw}_n)(x,Y^{[n]}) \rightarrow (qw)(x)$ $\mathbb{P}$-a.s.~by the strong law of large numbers. Therefore, by dominated convergence, it holds that $\alpha_n(x,y)\rightarrow \alpha_\infty(x,y)$ for all $x,y\in\Xspace$. Furthermore, by Theorem 1 in \cite{gagnon2023}, $\widetilde{q}^w_n(x,\cdot)\rightarrow q^w(x,\cdot)$ in sense of total variation. Thus,
\begin{align*}
	&|\alphid(x)-\alphmtm{n}(x)| \\
    &\leq \left| \int  (\widetilde{q}^w_n(x,\dif y)-q^w(x,\dif y))\alpha_n(x,y) \right| + \left| \int q^w(x,\dif y)(\alpha_n(x,y)-\alpha_\infty(x,y)) \right| \\
	&\leq \norm{\widetilde{q}^w_n(x,\cdot)-q^w(x,\cdot))}_{tv} + \left| \int q^w(x,\dif y)(\alpha_n(x,y)-\alpha_\infty(x,y)) \right|
\end{align*}
and the  claim follows by dominated convergence.
\qed
\subsection{Proof of Lemma \ref{lemma:alphidlb}} \label{app:alphidlb}
Consider \eqref{eq:alphainf} and let $\rho:=2/(2+\sigma^2)$.
Since $\psi$ is $L:=\sigma^2/(4(2+\sigma^2))$-smooth, and because $\rho\in(0,1)$, we obtain that
\begin{equation*}
    \psi(y+x\rho) \leq \psi(x\rho)+\iprod{\nabla \psi(x\rho)}{y}+ \frac{L}{2}\norm{y}^2 \leq \psi(x)+\iprod{\nabla \psi(x\rho)}{y}+ \frac{L}{2}\norm{y}^2.
\end{equation*}
Making the substitution $y\mapsto y-x\rho$, lower-bound the ideal acceptance rate as
\begin{align*}
	\alpha_\infty(x) &= \int \cal{N}\left(\dif y;0,I_d \cdot\sigma^2 \rho\right) \min\{1,\exp(-\psi(y+x\rho)+\psi(x))\} \\
	&\geq \int \cal{N}\left(\dif y;0,I_d \cdot\sigma^2 \rho\right) \min\left(1,\exp\left(-\iprod{\nabla \psi(x)}{y}- \frac{L}{2}\norm{y}^2\right)\right) \\
	&\geq  \int \cal{N}\left(\dif y;0,I_d \cdot\sigma^2 \rho\right)\exp\left(- \frac{L}{2}\norm{y}^2\right) \min\left(1,\exp\left(-\iprod{\nabla \psi(x)}{y}\right)\right) \\
	&= \int \cal{N}\left(\dif y;0,I_d \cdot\sigma^2 \rho\right)\exp\left(- \frac{L}{2}\norm{y}^2\right) \min\left(1,\exp\left(+\iprod{\nabla \psi(x)}{y}\right)\right)
\end{align*}
where in the third line we used the inequality $\min\{1,ab\}\geq \min\{1,a\}\min\{1,b\}$ and in the last equality the substitution $y\mapsto -y$. Averaging the last two expressions, using $\min\{1,a\}+\min\{1,1/a\}\geq 1$, and then applying Jensen's inequality to the exponential function, we then see that
\begin{align*}
    \alpha_\infty(x) &\geq \frac{1}{2} \int \cal{N}\left(\dif y;0,I_d \cdot\sigma^2 \rho\right)\exp\left(- \frac{L}{2}\norm{y}^2\right) \geq \frac{1}{2}\exp\left(-\frac{d\sigma^4}{4(2+\sigma^2)^2}\right).
\end{align*}
Noting that $(2+\sigma^2)\leq2\max\{1,\sigma^2\}$, we deduce the more user-friendly estimate
\begin{equation*}
	\alpha_\infty(x) \geq \frac{1}{2}\exp\left(-\frac{d\sigma^4}{16\max\{1,\sigma^4\}}\right) = \frac{1}{2}\exp\left(-\frac{d\min\{1,\sigma^4\}}{16}\right) 
\end{equation*} 
The final bound is then obtained by letting $\sigma=\zeta\cdot d^{-1/4}$ and $d\rightarrow\infty$. 
\qed
\end{appendix}

\end{document}